\documentclass[twocolumn,aps,pra,superscriptaddress]{revtex4-1}

\usepackage{times}
\usepackage{graphicx}
\usepackage{epsfig}
\usepackage{amsfonts}
\usepackage{amsmath}
\usepackage{amssymb}
\usepackage{color}
\usepackage{multirow}
\usepackage[colorlinks=true,citecolor=blue,urlcolor=blue]{hyperref}
\usepackage{float}

\usepackage[labelformat=simple]{subcaption}

\usepackage{bbm,bm}

\def\begeq{\begin{equation}}
\def\endeq{\end{equation}}
\def\begp{\begin{proposition}}
\def\endp{\end{proposition}}
\def\begl{\begin{lemma}}
\def\endl{\end{lemma}}
\def\begt{\begin{theorem}}
\def\endt{\end{theorem}}
\def\begd{\begin{definition}}
\def\endd{\end{definition}}

\newtheorem{lemma}{Lemma}[section]
\newtheorem{corollary}{Corollary}[section]

\newtheorem{theorem}{Theorem}
\newtheorem{proposition}{Proposition}[section]
\newtheorem{definition}{Definition}[section]

\newtheorem{proof}{Proof}[section]

\newcommand{\la}{\langle}
\newcommand{\ra}{\rangle}

\newcommand{\tr}{\text{tr}}

\newcommand{\ii}{\mathbf{i}}

\newcommand{\SP}{\text{ }}

\newcommand{\nl}{\newline}
\newcommand{\rw}{\rightarrow}
\newcommand{\ONE}{\mathbbm{1}}
\newcommand{\RR}{\mathbbm{R}}
\newcommand{\CC}{\mathbbm{C}}

\newcommand{\PP}{\mathbbm{P}}

\newcommand{\ee}{\mathcal{E}}
\newcommand{\dd}{\mathcal{D}}

\begin{document}


\title{Certifying an irreducible 1024-dimensional photonic state using refined dimension witnesses}



\date{\today}

\author{Edgar A. Aguilar}
\email[These authors contributed equally to this work.]{}
\affiliation{Institute of Theoretical Physics and Astrophysics, National Quantum Information Centre, Faculty of Mathematics, Physics and Informatics, University of Gdansk, 80-952 Gdansk, Poland}

\author{M\'at\'e Farkas}
\email[These authors contributed equally to this work.]{}
\affiliation{Institute of Theoretical Physics and Astrophysics, National Quantum Information Centre, Faculty of Mathematics, Physics and Informatics, University of Gdansk, 80-952 Gdansk, Poland}

\author{Daniel Mart\'inez}
\affiliation{Departamento de F\'{\i}sica, Universidad de Concepci\'on, 160-C Concepci\'on, Chile}
\affiliation{Millennium Institute for Research in Optics, Universidad de Concepci\'on, 160-C Concepci\'on, Chile}

\author{Mat\'ias Alvarado}
\affiliation{Departamento de F\'{\i}sica, Universidad de Concepci\'on, 160-C Concepci\'on, Chile}
\affiliation{Millennium Institute for Research in Optics, Universidad de Concepci\'on, 160-C Concepci\'on, Chile}

\author{Jaime Cari\~{n}e}
\affiliation{Departamento de F\'{\i}sica, Universidad de Concepci\'on, 160-C Concepci\'on, Chile}
\affiliation{Millennium Institute for Research in Optics, Universidad de Concepci\'on, 160-C Concepci\'on, Chile}

\author{Guilherme B. Xavier}
\affiliation{Millennium Institute for Research in Optics, Universidad de Concepci\'on, 160-C Concepci\'on, Chile}
\affiliation{Departamento de Ingenier\'ia El\'ectrica, Universidad de Concepci\'on, 160-C Concepci\'on, Chile}
\affiliation{Institutionen f\"{o}r Systemteknik, Link\"{o}pings Universitet, 581 83 Link\"{o}ping, Sweden}

\author{Johanna F. Barra}
\affiliation{Departamento de F\'{\i}sica, Universidad de Concepci\'on, 160-C Concepci\'on, Chile}
\affiliation{Millennium Institute for Research in Optics, Universidad de Concepci\'on, 160-C Concepci\'on, Chile}

\author{Gustavo Ca\~{n}as}
\affiliation{Departamento de F\'{\i}sica, Universidad del Bio-Bio, Av. Collao 1202, Concepci\'on, Chile}

\author{Marcin Paw{\l}owski}
\affiliation{Institute of Theoretical Physics and Astrophysics, National Quantum Information Centre, Faculty of Mathematics, Physics and Informatics, University of Gdansk, 80-952 Gdansk, Poland}

\author{Gustavo Lima}
\affiliation{Departamento de F\'{\i}sica, Universidad de Concepci\'on, 160-C Concepci\'on, Chile}
\affiliation{Millennium Institute for Research in Optics, Universidad de Concepci\'on, 160-C Concepci\'on, Chile}


\begin{abstract}
We report on a new class of dimension witnesses, based on quantum random access codes, which are a function of the recorded statistics and that have different bounds for all possible decompositions of a high-dimensional physical system. Thus, it certifies the dimension of the system and has the new distinct feature of identifying whether the high-dimensional system is decomposable in terms of lower dimensional subsystems. To demonstrate the practicability of this technique we used it to experimentally certify the generation of an irreducible 1024-dimensional photonic quantum state. Therefore, certifying that the state is not multipartite or encoded using non-coupled different degrees of freedom of a single photon. Our protocol should find applications in a broad class of modern quantum information experiments addressing the generation of high-dimensional quantum systems, where quantum tomography may become intractable.
\end{abstract}


\pacs{42.50.Xa,42.50.Ex,03.65.Ta}


\maketitle


{\em Introduction.---} The dimension $d$ of physical systems is a fundamental property of any model, and its operational definition arguably reflects the evolution of physics itself. In quantum mechanics, it can be seen as a key resource for information processing since higher dimensional systems provide advantages in several protocols of quantum computation \cite{ReviewQComp} and quantum communications \cite{ReviewQComm}. In the field of quantum foundations, a recent proposal suggests that in order to understand and create macroscopic quantum states it will be necessary to take advantage of high-dimensional systems \cite{Gisin18}. Therefore, it is natural to understand why there is an growing strive to coherently control quantum systems of large dimensions \cite{DIDW14,QKD16,Gao2010,Dada2011,Kwiat2005,Fickler2012,Krenn2014,Yao2012,Blatt2011,Pan2016,Barends2014,Malik2016,Fickler2016}. Nonetheless, such new technological advances require the simultaneous development of practical methods to certify that the sources are truly producing the required quantum states. In principle, one can rely on the process of quantum tomography \cite{DL01,DML04,JKMW01,TNWM02,Lima10,Lima11,Dardo15}, but this approach quickly becomes intractable in higher dimensions as at least $d^2$ measurements are required \cite{wootters}.

To address this problem, the concept of dimension witness (DW) was introduced. The original idea was based on the violation of a particular Bell inequality \cite{BPAGMS}, but then extended to the more practical prepare-and-measure scenario \cite{GBHA}. In general, DWs are defined as linear functions of a few measurement outcome probabilities and have classical and quantum bounds defined for each considered dimension \cite{DIDW14,BPAGMS,GBHA,Ahrens2012,Hendrych2012,Brunner2013,Bowles2014}. Thus, they allow for the device-independent certification of the minimum dimension required to describe a given physical system, and can also infer if it is properly described by a coherent superposition of logical states. Nevertheless, these tests do not provide information about the composition of the system, which is crucial for high-dimensional quantum information processing. This point has been recently investigated by W. Cong et al. \cite{Scarani17}, where they introduced the concept of an irreducible dimension witness (IDW) to certify the presence of an irreducible 4-dimensional system. Specifically, their IDW distinguishes if the observed data is created by one pair of entangled ququarts, or two pairs of entangled qubits measured under sequential adaptive operations and classical communication.

Here we introduce a new class of DWs, namely gamut DWs, which certifies the dimension of the system and has the new distinct feature of identifying whether any high-dimensional quantum system is irreducible. It is based on quantum random access codes (QRACs), which is a communication task defined in a prepare-and-measure scenario \cite{Ambainis99}. To demonstrate the practicability of our new technique we experimentally certify the generation of an irreducible 1024-dimensional photonic quantum system encoded onto the transverse momentum of single photons transmitted over programmable diffractive optical devices \cite{QKD16,LimaPRL05,Lima10,Lima11,Dardo15,Lima09,MulticoreFibers}. To our knowledge, our work represents an increase of about two orders of magnitude to any reported experiment using path qudits. From the recorded data one observes a violation of the bounds associated to all possible decompositions of a 1024-dimensional quantum system, thus, certifying that the generated state is not encoded using non-coupled different degrees of freedom of a photon, e.g., polarization and momentum. Nonetheless, our method is broadly relevant and should also find applications in multipartite photonic scenarios and new platforms for the fast-growing field of experimental high-dimensional quantum information processing.

{\em Gamut dimension witness.---}
As stated earlier, the protocol we use in our main theorem is based on QRACs. Thus, we first give a
brief description (see e.g. \cite{Ambainis99} for more details) of this task (see Fig. \ref{fig1}): one of the parties, Alice, receives two input dits: $x_1$ and $x_2\in\{1,\ldots,d\}$. She is then allowed to send one $d$-dimensional (quantum) state,
$\rho_{x_1x_2}$ to Bob, depending on her input. Bob is then given a bit $y\in\{1,2\}$ and his task is to guess
$x_y$. He does so by performing a quantum measurement $M^y$ and a classical
post-processing function $\mathcal{D}^y$. As a result, he outputs $b\in\{1,
\ldots,d\}$.

\begin{figure}
  \centering
  \includegraphics[width=0.45 \textwidth]{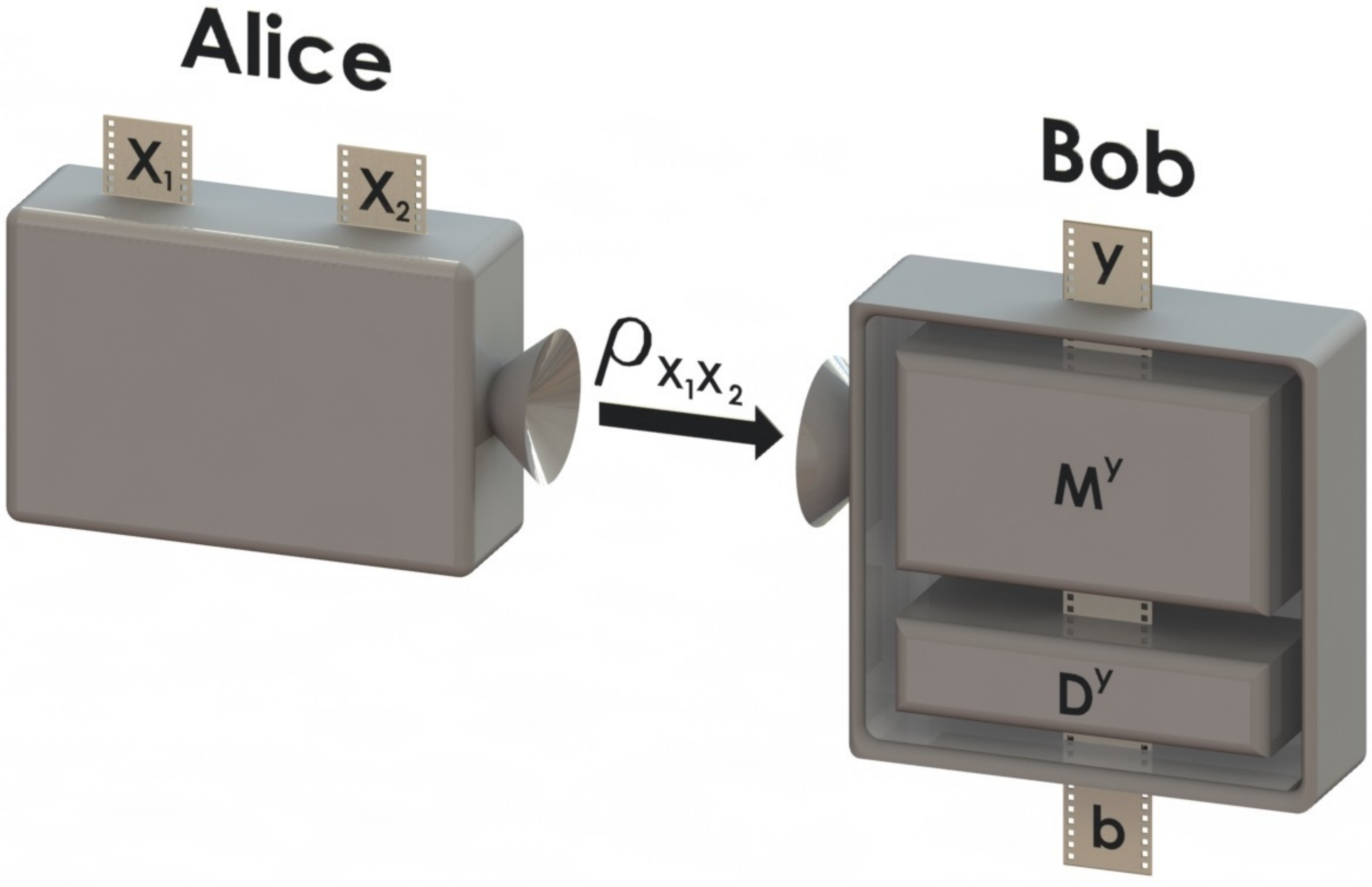}
  \caption{Our d-dimensional QRACs scenario. Alice receives the input dits $x_1$ and $x_2\in\{1,\ldots,d\}$, and prepares the state $\rho_{x_1x_2}$ which is sent to Bob. He receives the input $y\in\{1,2\}$, which defines the quantum measurement $M^y$ and the classical post-processing function $\dd^y$ to be applied to $\rho_{x_1x_2}$. As a result, Bob outputs $b$.}\label{fig1}
\end{figure}

For a single round of the protocol, the success probability is $\mathbb{P}
(b=x_y~|~x_1,x_2,y)$. As a figure or merit over many rounds with uniformly random
inputs, we employ the \textit{average success probability} (ASP): $\bar{p}=\frac{1}{2d^2}\sum_{x_1,x_2,y}\mathbb{P}(b=x_y~|~x_1,x_2,y).$
Thus, we are looking for the maximal value of $\bar{p}$, optimizing over all
possible encoding and decoding strategies.
It was proven \cite{Czechlewski18} that for classical strategies (i.e. classical states
and decoding functions), the optimal ASP is $\bar{p}
_{C_d} =\frac12(1+\frac1d)$. In the quantum case, the optimal strategy is
reached by using mutually unbiased bases (MUBs) for encoding and decoding
\cite{ABMP,Farkas18}, and the ASP is $\bar{p}_{Q_d}=\frac12(1+\frac{1}{\sqrt{d}})$.

Now, we estimate the optimal ASPs for composite systems, for all possible product structures, defined as follows:
\begin{definition}\label{def:product}
For a fixed $d$, we define a \textit{product structure} by the set $\big\{r,
\{d_k\},\{\alpha_k\}\big\}$. For a composite system, $d=\prod_{k=1}^r d_k$, where $d_k$ is the dimension of each subsystem and $r$ is the number of subsystems. The state of the composite system can be written as $\rho=\rho_{\alpha_1}^1\otimes\rho_
{\alpha_2}^2\otimes\cdots\otimes\rho_{\alpha_r}^r$. Here, $\alpha_k=
\text{c}$ and $\alpha_k=\text{q}$ are used to denote the ``classical'' and ``quantum'' nature of the subsystem, respectively. Then, $\rho_c^k\in\Delta_{d_k-1}$ is a classical state, and $\rho_q^k\in\mathcal{S}
(\mathbb{C}^{d_k})$ is a quantum state.
\end{definition}

Consider a set of measurement and state preparation settings, and
fix the total dimension of the physical system in question. We call a linear
function on the measurement outcome probabilities a gamut dimension witness (GDW), if
its extremal values for all possible product structures are different. For example, in $d=4$, a GDW has different extremal values for a ququart, two qubits, one qubit and a bit,
and a quart. The main theoretical result of this work is to demonstrate that d-dimensional QRACs can be used as GDWs for d-dimensional physical systems. To highlight this, we set it as a theorem.

\begin{theorem}
d-dimensional QRACs serve as gamut dimension witnesses using the ASP function.
\end{theorem}

The proof of this theorem and all related lemmas can be found
in the Supplemental Material \cite{SupMat}, which includes Refs.
\cite{ALMO,THMB}. Let us now sketch the
main tools for proving the theorem. They help to understanding the
problem, and can be independently used. Note that the following lemmas apply in more general QRAC scenarios as well \cite{SupMat}.

We assume that Bob's measurements have the same product structure as the state generated by Alice. That is, we exclude that Bob's state certification would use entangling measurements. The motivation here is to rule out sequential uses of lower dimensional systems as a way to simulate higher dimensional statistics, e.g.\ to discriminate between $n$ sequential uses of a $d$ dimensional system, and a $d^n$ dimensional system. A physical motivation for this assumption is to think that if Alice cannot couple a particular set of degrees of freedom (e.g.\ polarization and momentum), then neither can Bob because he has access to the same equipment as Alice does \footnote{We also note that numerical evidence on small dimensional examples suggests that if the quantum states are in a product form, then entangling measurements do not improve the ASP.}.

Therefore, the most general strategy for decoding the $d$-dimensional system $\rho=\rho^1\otimes\rho^2\otimes
\cdots\otimes\rho^r$ is as follows: Bob performs sequential adaptive measures on the subsystems in the sense of \cite{Scarani17}. He starts by measuring subsystem $\rho^1$ to obtain the outcome $b^1$. Then, his choice of the measurement to be performed in $\rho^2$ may depend on $b^1$. Successively, each measurement on $\rho^k$ can depend on all the measurement outcomes obtained previously. After performing all measurements, Bob feeds
the obtained outcomes to a classical post-processing function, and outputs his
final guess on $x_y$, which is $b=\mathcal{D}^y(b^1b^2\ldots b^r)$.

The bounds of the GDW in this general scenario are extremely hard to obtain. The following results
help making the analysis easier. First, it is argued in \cite{Ambainis99} that
in an optimal strategy, it is enough to use encoded pure states. Similarly, it has been shown that rank 1 projective measurements (explicitly: mutually unbiased bases) optimize two-input QRACs \cite{Farkas18}. Thus,in the following we only deal with pure states for both Alice and Bob. Additionally,
we can eliminate classical post-processing functions:
\begin{lemma}\label{lem:identity}
In QRACs, for optimality of the ASP, there is no need for classical
post-processing functions.
\end{lemma}
Last, we note that:
\begin{lemma}\label{lem:nonadaptive}
In QRACs, for optimality of the ASP, there is no need for sequential adaptive measurements.
\end{lemma}

Observe that the above lemmas together imply that the highest ASP for a composite system can be achieved with a strategy that consists of $r$ QRACs \textit{in parallel}, one on each subsystem $\rho^k$, independently. In this case, if we write Alice's inputs as
dit-strings $x_y=x^1_yx^2_y\ldots x^r_y$, the success probability for each round is: $\mathbb{P}(b=x_y|x_1,x_2,y) = \prod_{k = 1}^r \mathbb{P}(b^k=x_y^k|x_1^k,x_2^k,y)$. The optimal $\bar{p}$ is not necessarily given by the independent optimal strategies on the individual subspaces.
Therefore, in order to optimize it we introduce the \textit{trade-off function}
$\mathcal{M}_d(z)$ (see the Supplemental Material \cite{SupMat}, which includes
Ref. \cite{inequalities}), which provides the optimal
probability of guessing dit $x_2$ given a fixed probability of guessing dit
$x_1$. Let $z=\mathbb{P}(\mbox{\text{Bob correctly guesses}} \, x_1)$. Then, $\mathcal{M}_d(z)$ in dimension $d$ is defined by $\mathcal{M}_d(z) = \max \{\mathbb{P}(\mbox{\text{Bob correctly guesses}} \, x_2)|z\}$, where the maximization is limited to all encoding-decoding strategies respecting
the condition of guessing $x_1$ with probability $z$. Thus, in a general case
\begin{equation}
\bar{p}_{Q_{d_1}\ldots C_{d_r}} = \underset{z^1,\ldots,z^r}{\max}\frac{1}{2}[z^1\cdots z^r + \mathcal{M}_{d_1}^q(z^1)\cdots\mathcal{M}_{d_r}^c(z^r)],
\label{eq:asp_tof}
\end{equation}
where we denote $d$-dimensional quantum and classical states by $Q_d$ and $C_d$,
respectively. $\mathcal{M}^q_d$ and $\mathcal{M}^c_d$ are the corresponding
quantum, and classical trade-off functions \cite{SupMat}. Therefore, $\bar{p}$ is a function
of $r$ real variables, and its maximum can be found using standard heuristic numerical search algorithms \cite{Press}. We present the ASP optimal values for some relevant cases of a $d=1024$ dimensional system in Table \ref{table1}. The full list of cases is found in the Supplemental Material \cite{SupMat}. Note that the gaps between the different ASP values
are large enough to be experimentally observed, as we demonstrate next.

\begin{table}
\centering
\begin{tabular}{c@{\hspace{2cm}}c}\hline \hline
Case & Optimal $\bar{p}$ \\ \hline
$Q_{1024}$ & 0.515625 \\
$Q_{512}Q_2$ & 0.500980 \\
$Q_{512}C_2$ & 0.500973 \\
$Q_{32}Q_{32}$ & 0.500521 \\
$(Q_2)^{10}$ & 0.500493 \\
$Q_2C_{512}$ & 0.500489 \\
$C_{1024}$ & 0.500488 \\ \hline \hline
\end{tabular}
\caption{Relevant cases for a 1024-dimensional system and the
respective optimal ASPs (Eq.(\ref{eq:asp_tof})) considering each product structure. The full table can be found in the Supplemental Material \cite{SupMat}.}\label{table1}
\end{table}

{\em Experiment.---} To demonstrate the practicability of our technique we generate a 1024-dimensional photonic state, encoded into the linear transverse momentum of single-photons, and use the 1024-dimensional QRAC GDW to certify that it is an irreducible quantum system. To achieve this, we first show that the ASP (Eq.(\ref{eq:asp_tof})) can be written as a simple function of the detection events. Then, we observe that our recorded statistics violate the second highest ASP bound, $Q_{512}Q_2$, given in Table \ref{table1}. Thus, ensuring that it is an irreducible 1024-dimensional quantum system.

In the 1024-dimensional QRAC GDW, Bob measures the elements of the two 1024-dimensional MUBs given in the Supplemental Material \cite{SupMat}. We denote the MUBs states by $|m_{j}^{y}\ra$, where $y = 1,2$ defines the measuring base $MUB_1$ or base $MUB_2$, and $j=1,...,1024$ denotes the state of a given base. Alice's state is written in terms of the two input dits $x_1$ and $x_2$ as an equal superposition of the states Bob would need to guess $x_y$ correctly:
\begin{equation}
\label{eq:optstate}
  |\Psi_{x_1x_2}\ra = \frac{1}{N}(|m_{x_1}^{1}\ra + \text{sgn}(\la m_{x_1}^{1}|m_{x_2}^{2}\ra)|m_{x_2}^{2}\ra),
\end{equation}
where $N = \sqrt{2(1+\frac{1}{32})}$ is a normalization factor and sgn is the sign function. The optimality of the encoded states \eqref{eq:optstate}, and the use of MUBs is derived in the Supplemental Material \cite{SupMat}.

For the experimental test, we resort to the setup depicted in Fig.~\ref{fig:setup}. At the state preparation block, the single-photon regime is achieved by heavily attenuating optical pulses with well calibrated attenuators. An acousto-optical modulator (AOM) placed at the output of a continuous-wave laser operating at 690nm is used to generate the optical pulses. The average number of photons per pulse is set to $\mu=0.4$. In this case, the probability of having non-null pulses is $P(n\geq1|\mu=0.4)=33\%$. Pulses containing only one photon are the majority of the non-null pulses generated and accounts to 82$\%$ of the experimental runs. Thus, our source is a good approximation to a non-deterministic single-photon source, which is commonly adopted in quantum communications \cite{ReviewQComm}.

\begin{figure}
  \centering
  \includegraphics[width=0.5 \textwidth]{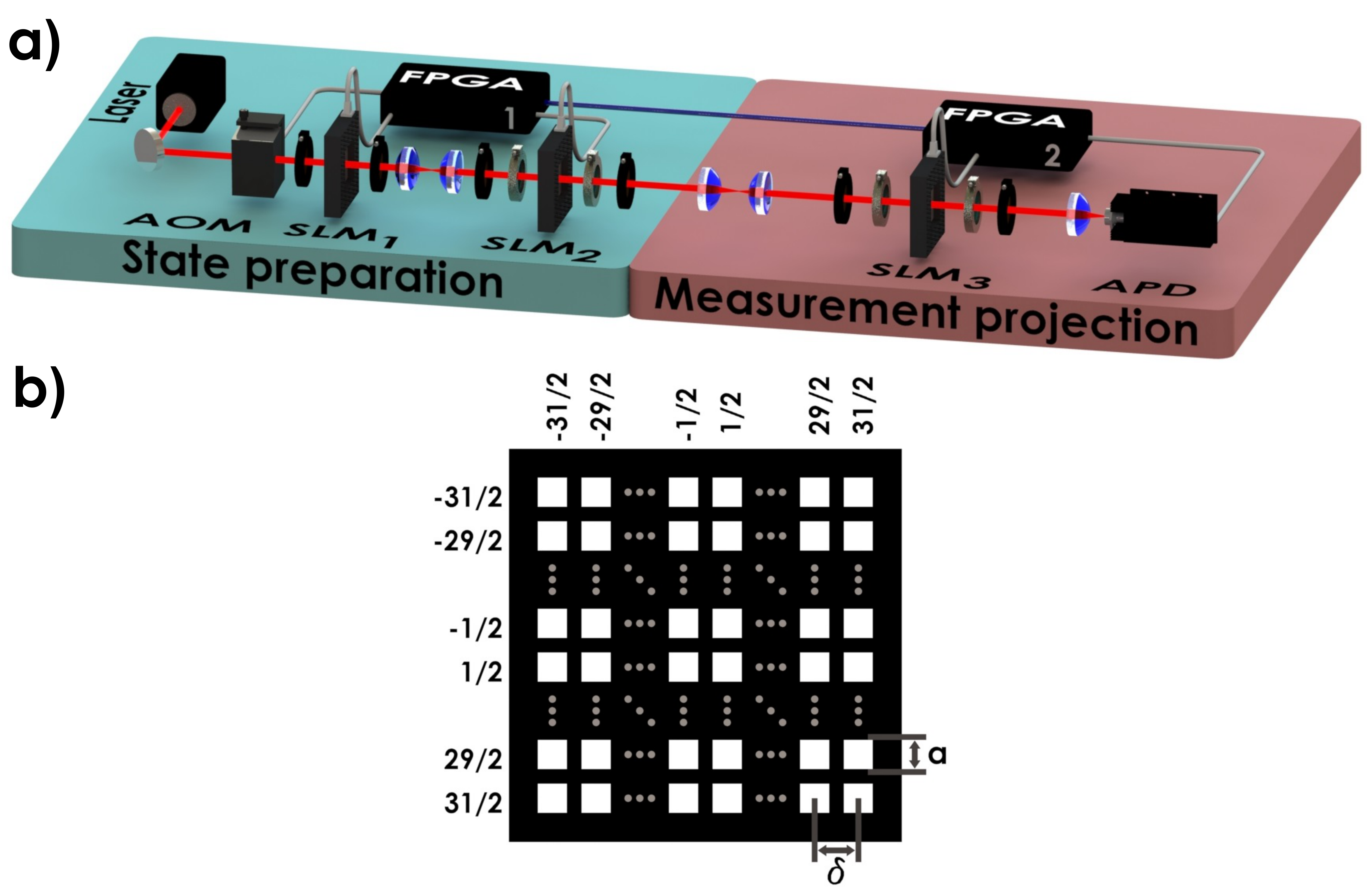}
  \caption{\textbf{a}. Experimental setup. We employ a prepare-and-measure scheme to generate and project spatial qudits, encoded into the linear tranverse momentum of single-photons. At the state preparation block, the spatial encoding is applied through two spatial light modulators (SLMs), and the state projection is likewise performed by a SLM combined with a point-like single-photon detector (APD) at the measurement projection block (see main text for details). \textbf{b}. The 32$\times$32-square mask addressed by the SLMs.}\label{fig:setup}
\end{figure}

The single-photons are then sent through two spatial light modulators, SLM1 and SLM2, addressing an array of 32$\times$32 transmissive squares. The square side is $a=96\mu$m and they are equally separated by $\delta=160\mu$m (see Fig. \ref{fig:setup}b). Thus, effectively creating a 1024-dimensional quantum state defined in terms of the number of modes available for the photon transmission over the SLMs \cite{QKD16,LimaPRL05,Lima10,Lima11,Dardo15,Lima09}. Specifically, the state of the transmitted photon is given by $|\Psi\rangle = \frac{1}{\sqrt{C}}\sum_{l=-l_{N_c}}^{l_{N_c}}\sum_{v=-l_{N_r}}^{l_{N_r}}\sqrt{t_{lv}}e^{-i\phi_{lv}}|c_{lv}\rangle,$ where $|c_{lv}\rangle$ is the logical state representing the photon transmitted by the ($l,v$) square. $t_{lv}$ represents the transmission and $\phi_{lv}$ the phase-shift given by the ($l,v$) square. The transmission of each square is controlled by the SLM1, which is configured for amplitude-only modulation. The phases $\phi_{lv}$ are controlled by SLM2 working on the configuration of phase-only modulation \cite{Lima11}. $N_c$ and $N_r$ represent the number of columns and rows, respectively. For simplicity, we define $l_{N_c} \equiv \frac{N_c-1}{2}$, $l_{N_r}\equiv\frac{N_r-1}{2}$, and $C$ is the normalization factor.

At the measurement block we use a similar scheme to the one used in the state preparation block. It consists of a SLM3, also configured for phase-modulation, and a ``pointlike'' avalanche single-photon detector (APD). As explained in details at \cite{Lima11,QKD16}, by placing the ``pointlike'' APD at the SLM3 far-field (FF) plane, and properly adjusting the ($l,v$) square phase-shifts, Bob can detect any state $|m_{j}^{y}\ra$ required for the 1024-dimensional QRAC session. The ``pointlike'' APD is composed of a pinhole (aperture of $10\mu$m diameter) fixed at the center of the FF plane, followed by the APD module. In this case, the probability of photon detection is proportional to the overlap between the prepared and detected states. For the case of a $d$-dimensional QRACs implemented with a single-detector scheme, we show at the Supplemental Material (see 
\cite{SupMat} and Refs. \cite{DIDW14,QKD16,Pan2016,Fickler2012} therein) that the ASP function can be written as
\begin{equation}
\bar{p} = \frac{D_1}{D_1+D_2}.
\label{eq:pDI}
\end{equation} We first consider the events with $x_y=j$ (again, $j=1,...,1024$ denotes the state of a given base) and define the total number of such events to be $X_1$. Then, we define $D_1$ as the number of "clicks" recorded in the experiment in those cases. Likewise, we denote $X_2$ to be the number of events where $x_y \neq j$ and define $D_2$ to be the clicks in those cases.

\begin{figure}
  \centering
  \includegraphics[width=0.45 \textwidth]{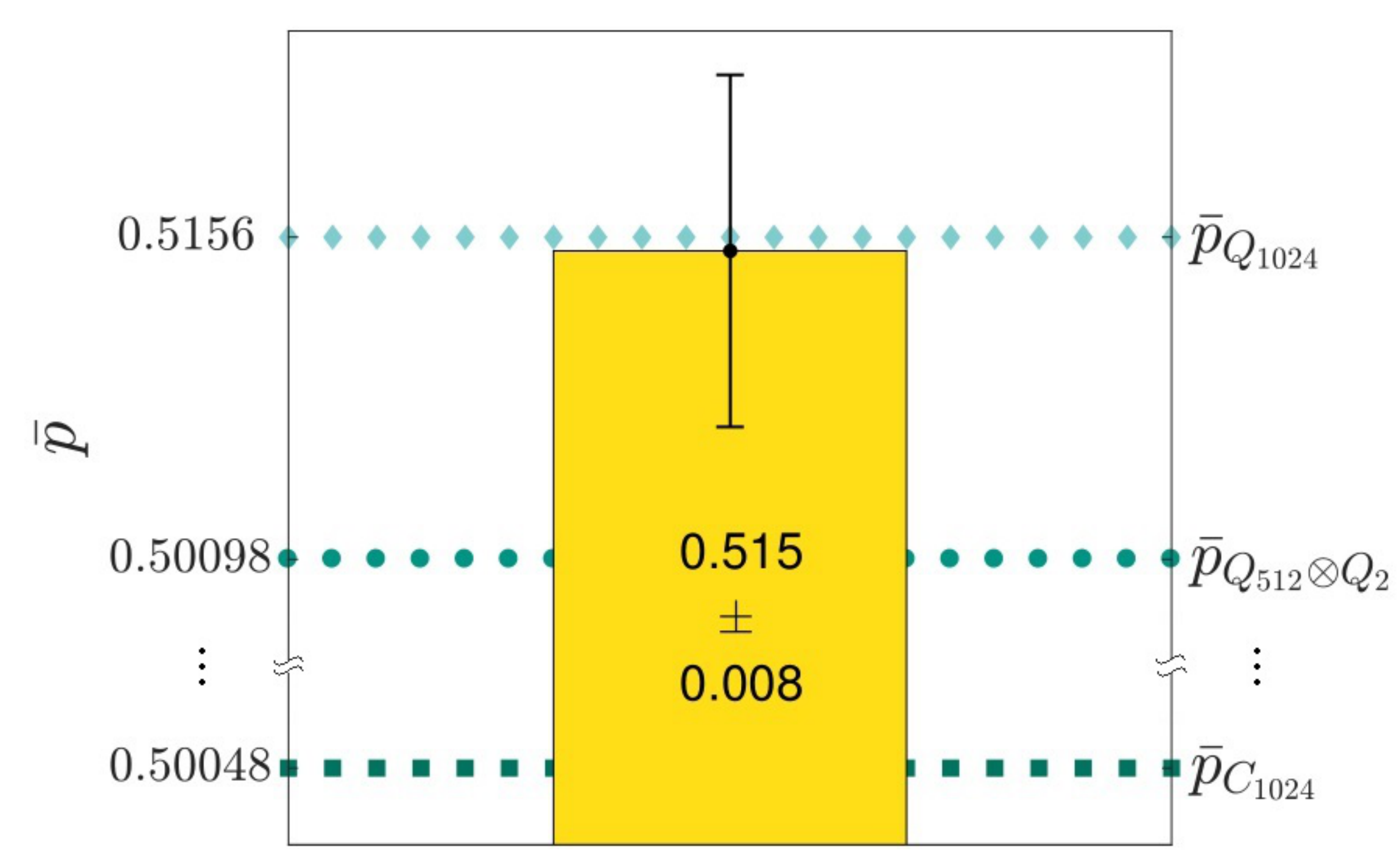}
  \caption{Experimental results. We experimentally observe $\bar{p}=0.515\pm 0.008$, violating the second highest ASP bound $\bar{p}_{Q_{512}\otimes Q_2}$ (see Tab.\ref{table1}). The error bar is calculated assuming Poissonian statistics for a photon detection event.}\label{fig:psucc}
\end{figure}

By means of two field-programmable gate arrays (FPGA) electronic modules we are able to automate and actively control both blocks of the setup. At the state preparation block, since the state $|\Psi\rangle$ needs to be randomly selected from the set of states defined by the 1024-dimensional QRACs, a random number generator (QRNG - Quantis) is connected to FPGA1. FPGA1 controls the optical pulse production rate by the AOM, set at 60 Hz as limited by the refresh rate of the SLMs. Each attenuated optical pulse corresponds to an experimental round. At the measurement block, a second QRNG is connected to FPGA2, providing an independent and random selection for the projection $|m_{j}^{y}\ra$ at each round. FPGA2 also records whether a detection event occurs. The overall detection efficiency is 13\%. The protocol is executed as follows: In each round, FPGA1 reads the dits $x_1$ and $x_2$ produced by its QRNG. Then, FPGA1 calculates the amplitude and phase of each ($l,v$) square of SLM1 and SLM2 to encode the state $|\Psi_{x_1x_2}\ra$ onto the spatial profile of the single-photon in that experimental round. Simultaneously, FPGA2, reads from its QRNG the value of $y$ and $j$. Similar to what is done in the state preparation block, FPGA2 also calculates the phase for each ($l,v$) square in SLM3 to implement the chosen projection $|m_{j}^{y}\ra$. The amplitude and relative phase for each SLM was previously characterized in order to obtain the modulation curves as a function of its grey level. In this experiment, this is necessary to dynamically generate all possible states, as it would be unfeasible to pre-record pre-defined masks for the SLMs on the FPGAs for each one of the $1024^2$ required initial states.

The experiment continuously ran over 316 hours. In this way, the statistics fluctuations observed for $D_1$ and $D_2$ were sufficiently small to unambiguously certify the generation of an irreducible 1024-dimensional quantum system. The overall visibility in our system is $97.00 \pm 0.07\%$ and the corresponding recorded average success probability is $\bar{p}=0.515\pm 0.008$. In Fig.~\ref{fig:psucc} we compare it with the second highest ASP bound shown in table \ref{table1}, associated with a composite system of the type $Q_{512}Q_2$. This certifies, only from the statistics recorded, that the generated state is not encoded using non-coupled different degrees of freedom of a photon, for instance polarization and momentum. Thus, ensuring it to be an irreducible 1024-dimensional quantum system that can provide all the advantages known for high-dimensional quantum information processing, in the sense explained in \cite{Scarani17}.

{\em Conclusion.---} Dimension witnesses are practical protocols on the field of quantum information as they allow one to obtain information regarding unknown quantum states \cite{BPAGMS, GBHA}. They are especially appealing while addressing the generation and characterization of high-dimensional quantum states, where quantum tomography demands at least $d^2$ measurements \cite{wootters}. In general, DWs are functions of only a few measurement outcome probabilities and allow for assessments on the dimension required to describe a given quantum state in a device-independent way \cite{DIDW14,BPAGMS,GBHA,Ahrens2012,Hendrych2012,Brunner2013,Bowles2014}. Here we give a step further by introducing a new class of DW, which certifies the dimension of the system, and has the new distinct feature of allowing the identification whether a high-dimensional system is irreducible. The application of this new feature is of broad relevance for several new architectures aiming for high-dimensional quantum information processing \cite{DIDW14,QKD16,Gao2010,Dada2011,Kwiat2005,Fickler2012,Krenn2014,Yao2012,Blatt2011,Pan2016,Barends2014,Malik2016,Fickler2016}, and the understanding of macroscopic quantumness \cite{Gisin18}. We demonstrate the practicability of our technique by using it to certify the generation of an irreducible 1024-dimensional photonic quantum state encoded into the linear transverse momentum of single-photons transmitted by programable diffractive apertures, which have been used for several high-dimensional quantum information processing tasks \cite{QKD16,MulticoreFibers,Marques15,KS14,MEQSD17}.

\begin{acknowledgments}
E.A.A. and M.F. thank Micha\l~Oszmaniec for fruitful discussions.
This work was supported by First~TEAM/2016-1/5, Sonata~UMO-2014/14/E/ST2/00020, Fondecyt~1160400, Fondecyt~11150324, Fondecyt~1150101 and Millennium Institute for Research in Optics, MIRO. E.A.A. acknowledges support from CONACyT. J.F.B. acknowledges support from Fondecyt~3170307. J.C. acknowledges support from Fondecyt~3170596. D.M. and M.A. \ acknowledge support from CONICYT.
\end{acknowledgments}


\onecolumngrid

\begin{center}
\LARGE{Supplemental Material}
\end{center}

The supplemental material is organized into two sections:  Theory (\ref{sec:theory}), and Experimental Considerations (\ref{sec:Exps}). 

The theoretical section makes all the formal definitions and provides the proofs of Theorem 1, Lemma 1, and Lemma 2 of the main text. We further clarify Equation (1) of the main text, as well as showing the explicit form of the trade-off functions. The theoretical section ends with two examples. In particular we calculate a table of all of the possible quantum partitions for $d=1024$ as direct proof that indeed: $Q_{1024} > Q_{512}Q_2 >$ ``all other partitions". (Table \ref{table1024})

The experimental section explicitly show the representation of the MUBs that were used in the experiment. We also formalize the single-detector scheme, and explain how the figure of merit (Equation (3) of the main text) is derived. Finally, we show how this figure of merit depends on the overall detection efficiency $\nu$ and average photon number per pulse $\mu$.

\section{Theory}\label{sec:theory}
\subsection{Formal Definitions and Problem Formulation}

We begin by defining $n^d\rw 1$ \textit{Random Access Codes} (RACs) rigorously. RACs is a strategy in which Alice tries to compress an $n$-dit string into $1$ dit, such that Bob can recover any of the $n$ dits with high probability \cite{ALMO}. Specifically, Alice receives an input string $x=x_1x_2\ldots x_{n}$ drawn from a uniform distribution, where $x_i\in [d]$, with $[d]=\{1,2,\ldots,d\}$. Note that in the special case of the main manuscript, we always use $x=x_1x_2$. She then uses an encoding function $\ee:[d]^n\rw [d]$, and is allowed to send one dit $a_x=\ee(x)$ to Bob. On the other side, Bob receives an input $y\in [n]$ (also uniformly distributed), and together with Alice's message $a_x$ uses one of $n$ decoding functions $\dd^y:[d]\rightarrow [d]$, to output $b=\dd^y(a_x)$ as a guess for $x_y$. If Bob's guess is correct (i.e. $b=x_y$) then we say that they \textit{win}, otherwise we say that they \textit{lose}. We can then quantify their probability of success $\PP(\dd^y(\ee(x))=x_y)$, which in general depends on their inputs and on the chosen \textit{strategy} $(\ee,\dd)$, where $\dd=\{\dd^y\}_{y=1}^n$.

Similarly, one defines the d-dimensional $n^d\rw 1$ Quantum Random Access Codes (QRACs) with the only change being that Alice tries to compress her input string into a $d$-dimensional quantum system (see Fig.\ref{fig:RAC}). Alice encodes her $n$-dit string via $\ee:[d]^n\rw \mathcal{S}(\CC^d)$, and sends the $d$-dimensional system $\rho_x = \ee(x)$ to Bob. He then performs some decoding to output his guess $b\in[d]$ for $x_y$. The decoding function is a quantum measurement followed by classical post-processing, as we clarify next.

\begin{figure}[h]
\begin{center}
\includegraphics[width=0.5\textwidth]{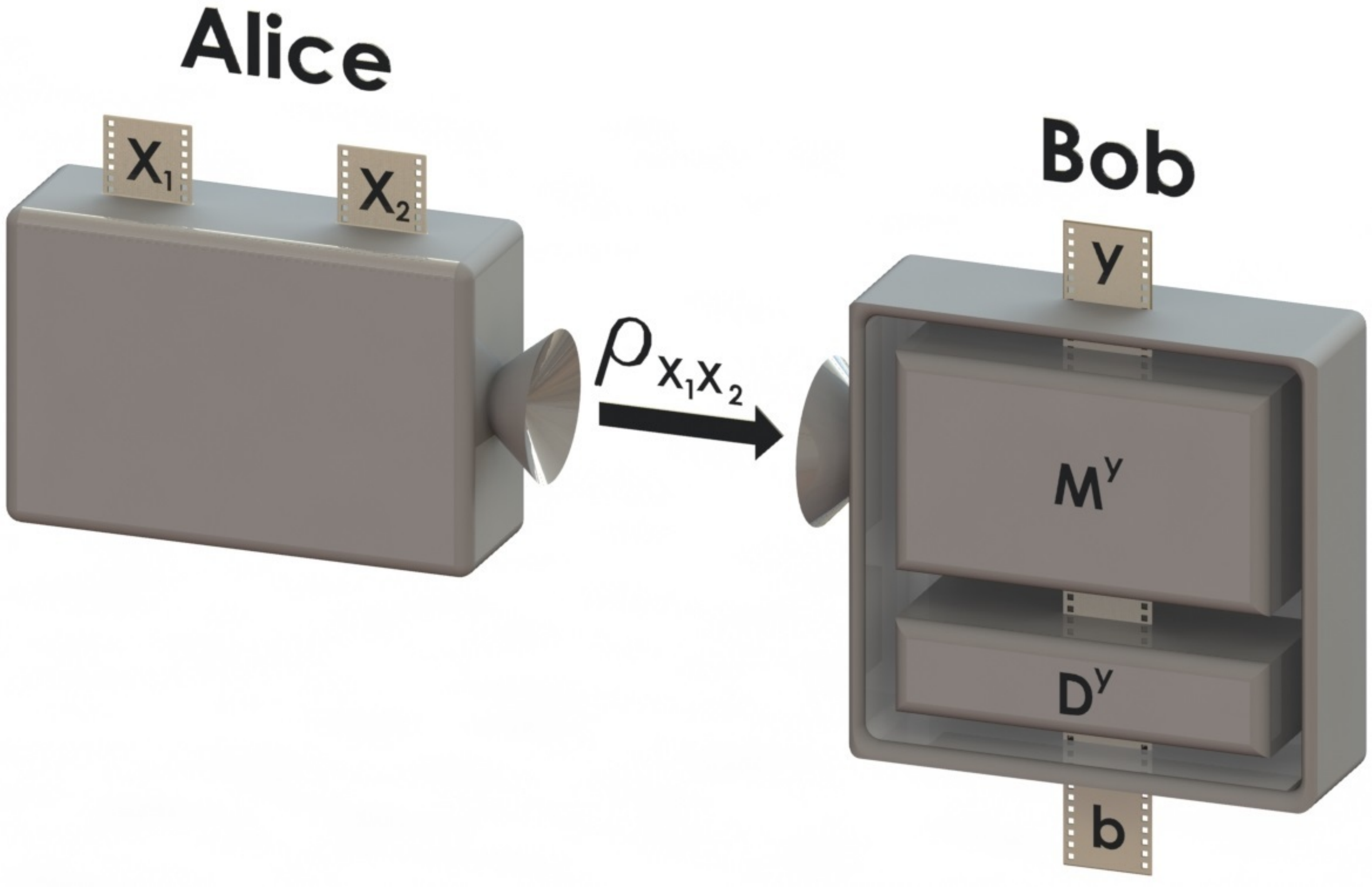}
\end{center}
\caption[]{d-dimensional $2^d\rw 1$ QRACs scenario, which is the one considered in the main manuscript. Alice receives the input dits $x_1$ and $x_2\in\{1,\ldots,d\}$, and prepares the state $\rho_{x_1x_2}$ which is sent to Bob. He receives the input $y\in\{1,2\}$, which defines the quantum measurement $M^y$ and the classical post-processing function $\dd^y$ to be applied to $\rho_{x_1x_2}$. As a result, Bob outputs $b$.}\label{fig:RAC}
\end{figure}

\begin{definition} {$~$}\nl
A \textbf{quantum decoding strategy} is $\dd=\big\{\{M_l^y\}_l,\dd^y\big\}_{y=1}^n$, i.e. $n$ pairs of measurement operators $\{M_l^y\}_l$ (normalized $\sum_l M_l^y = \ONE~ \forall y$ , and positive semi-definite $ M_l^y \geq 0 \SP \forall l,y$), and classical post-processing functions $\dd^y:[d]\rw[d]$, such that if Bob receives as input $\rho_{x}$ and $y$, he outputs $b=\dd^y(l)$ with probability $\tr[\rho_{x} M_l^y]$.
\end{definition}

To quantify the performance of a given encoding-decoding strategy, we shall employ the \textit{average success probability} (ASP) $\bar{p}$ as our figure of merit.

\begin{definition}{$~$}\nl \label{def:asp}
The \textbf{Average Success Probability} of a given encoding-decoding strategy $(\ee,\dd)$ is:
\begin{equation}\label{eq:asp}
\bar{p} = \frac{1}{nd^n} \sum_{x,y} \PP(B=x_y|X=x, Y=y),
\end{equation}
\end{definition} where uppercase letters $X,Y,B$ denote random variables, while the corresponding lowercase letters represent the events (i.e. the values the random variables can take). Another useful way of understanding the ASP is by viewing the whole QRAC protocol as a game and thinking of the ASP as the probability that Alice and Bob win any given round. Loosely speaking:
\begin{equation}\label{eq:altasp}
\bar{p} = \PP(B=\text{correct}).
\end{equation} Nonetheless, the real object of interest is the \textit{optimal average success probability}, which corresponds to the maximal value of $\bar{p}$ taken over all possible encoding-decoding strategies. Explicitly:
\begin{equation}
\bar{p}_{(C,Q)_d} = \max_{\{\ee,\dd\}} \frac{1}{nd^n} \sum_{x,y} \PP(B=x_y|X=x, Y=y),
\end{equation} with $C$ and $Q$ respectively representing the classical and quantum scenarios.

\begin{definition}{$~$}\nl \label{def:product}
For a fixed $d$, we define a \textit{product structure} by the set $\big\{r,
\{d_k\},\{\alpha_k\}\big\}$. For a composite system, $d=\prod_{k=1}^r d_k$, where $d_k$ is the dimension of each subsystem and $r$ is the number of subsystems. The state of the composite system can be written as $\rho=\rho_{\alpha_1}^1\otimes\rho_
{\alpha_2}^2\otimes\cdots\otimes\rho_{\alpha_r}^r$. Here, $\alpha_k=
\text{c}$ and $\alpha_k=\text{q}$, are used to denote the ``classical'' and ``quantum'' nature of the subsystem, respectively. Then, $\rho_c^k\in\Delta_{d_k-1}$ is a classical state, and $\rho_q^k\in\mathcal{S}
(\mathbb{C}^{d_k})$ is a quantum state.
\end{definition}

We are now in a position to formally pose the central question of this paper. Suppose Alice creates states of dimension $d$ with a certain \textit{product structure}, i.e. she creates the state $\rho=\rho_{\alpha_1}^1\otimes\rho_{\alpha_2}^2\otimes\cdots\otimes\rho_{\alpha_r}^r$. When dealing with separable states, it is easier to speak as if the information was encoded into distinct non-interacting physical systems. Of course it could equivalently be the case that there is only one physical system with non-interacting degrees of freedom creating the abstract separable structure, but for the sake of clarity we will keep the first picture in mind. This may be viewed as adding constraints to Alice's possible encoding functions $\ee$. 

We must further assume the same constraints on Bob's measurements. This might seem arbitrary, as we are only interested in the nature of the prepared state. Nevertheless, one can argue that if e.g. Bob is allowed to perform ``entangling" measurements, this device might as well be located in Alice's lab, allowing her to prepare an arbitrary entangled state which does not respect the original constraints. That is, we are interested in the scenario where both Alice and Bob have the same technological equipment at their disposal, as is the case in experiments \footnote{We also note that numerical evidence on small dimensional examples suggests that if the quantum states are in a product form, then entangling measurements do not improve the ASP.}. We remark that this assumption was also used to prove robustness in \cite{Scarani17}. Table \ref{tab:cases} gives an example of different product structures if $r\leq 2$.

\begin{table}[ht!]
\begin{center}
\begin{tabular}{ | c | c |  }
\hline
Case & Constraints on $\ee$, and $\dd $  \\ \hline
\multirow{2}{*}{$Q_{d_1d_2}$}
& Fully Quantum (No Constraints) \\
		& $\rho \in \mathcal{S}(\CC^{d}) $ \\ \hline
		\multirow{3}{*}{$Q_{d_1}Q_{d_2}$}
		& Separable Quantum States \\
			& $\rho = \rho_q^1 \otimes \rho_q^2 $ \\
			& $\rho_q^1 \in \mathcal{S}(\CC^{d_1})$ , $\rho_q^2 \in \mathcal{S}(\CC^{d_2})$ \\ \hline
			\multirow{3}{*}{$Q_{d_1}C_{d_2}$}
			& Classical Quantum \\
				& $\rho = \rho_q^1 \otimes \rho_c^2 $ \\
				&  $\rho_q^1 \in \mathcal{S}(\CC^{d_1})$, $\rho_c^2 \in \Delta_{d_2-1}$ \\ \hline
				\multirow{3}{*}{$C_{d_1}Q_{d_2}$}
				& Classical Quantum \\
					& $\rho = \rho_c^1 \otimes \rho_q^2$ \\
					&  $\rho_c^1 \in \Delta_{d_1-1}$, $\rho_q^2 \in \mathcal{S}(\CC^{d_2})$ \\ \hline
					\multirow{2}{*}{$C_{d_1d_2}$}
					& Classical \\
						& $\rho  \in \Delta_{d_1d_2-1} $ \\ \hline
\end{tabular}
\caption{Example of Alice's possible product structures, if the dimension $d=d_1d_2$ factorizes and $r\leq 2$. We assume that the measurement $\dd$ has the same product structure as the encoding $\ee$.}
\label{tab:cases}
\end{center}
\end{table}

Our main theorem states that the optimal ASPs of QRACs serve as a tool to differentiate these product structures. For convenience we also restate it here.

\setcounter{theorem}{0}
\begin{theorem}[Main theorem]
\label{thm:main}
d-dimensional $2^d \to 1$ QRACs serve as gamut dimension witnesses using the ASP function.
\end{theorem}

The rest of this section is dedicated to proving Theorem \ref{thm:main}. 

\subsection{Proofs of Lemmas 1 \& 2}

We will show how to transform from the most general setup from Fig. \ref{fig:simplifiedRAC_a}, into the setup of Fig. \ref{fig:simplifiedRAC_b}. In order to do this, we restrict the encoding function to only pure states (the optimality of which is demonstrated in Ref.\cite{ALMO}), the measurements to be projectives (shown optimal for our case in \cite{Farkas18}), and prove two lemmas that show that both (1) classical post-processing functions, and (2) sequential adaptive strategies, are all unnecessary on Bob's side. Note that these lemmas apply in the general $n^d\to1$ case.

\begin{figure}[!ht]
\centering
\begin{subfigure}{0.48\linewidth}
  \centering
\includegraphics[width=0.9\linewidth]{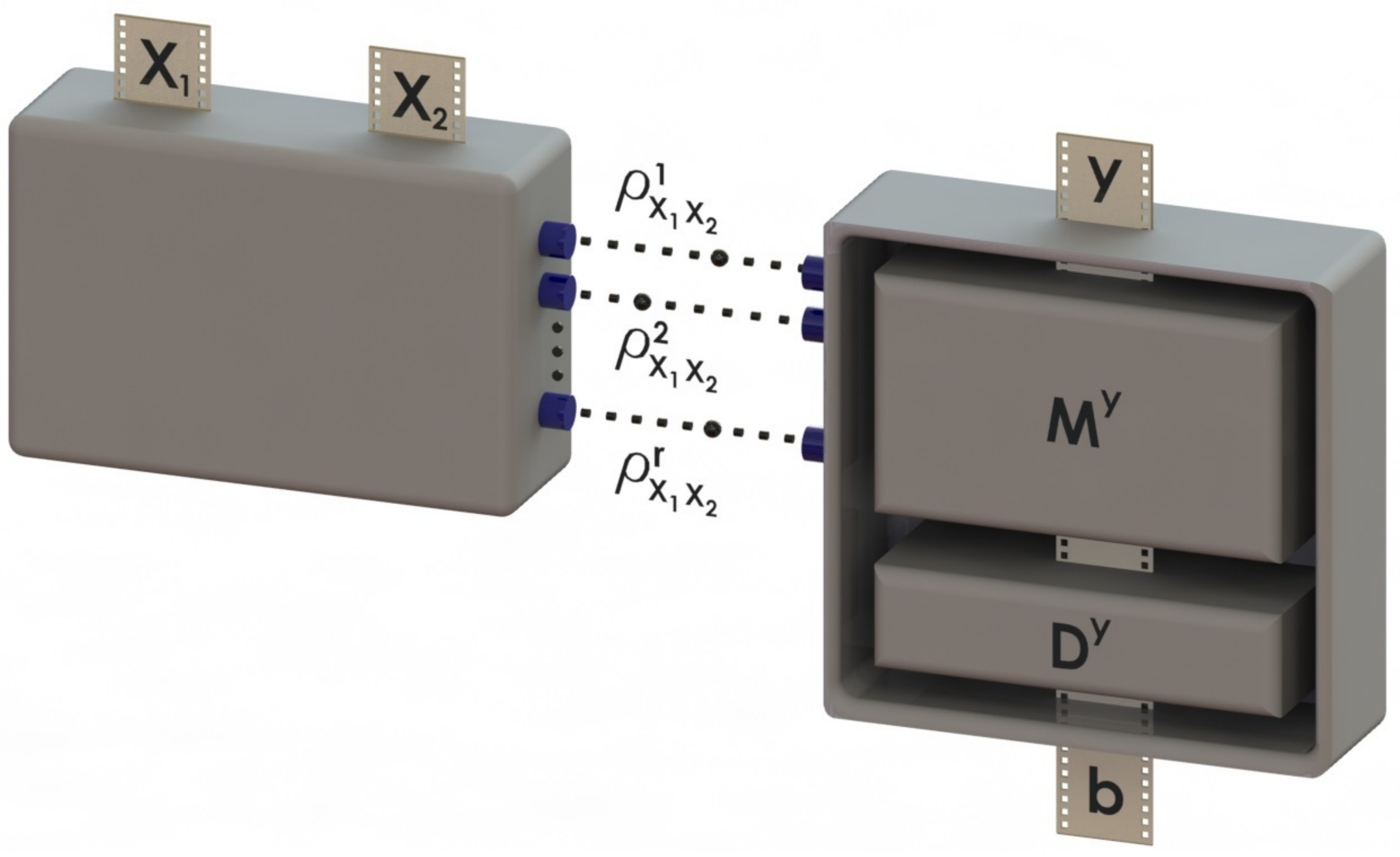}
\caption{}\label{fig:simplifiedRAC_a}
\end{subfigure}
\begin{subfigure}{0.48\linewidth}
  \centering
\includegraphics[width=0.9\linewidth]{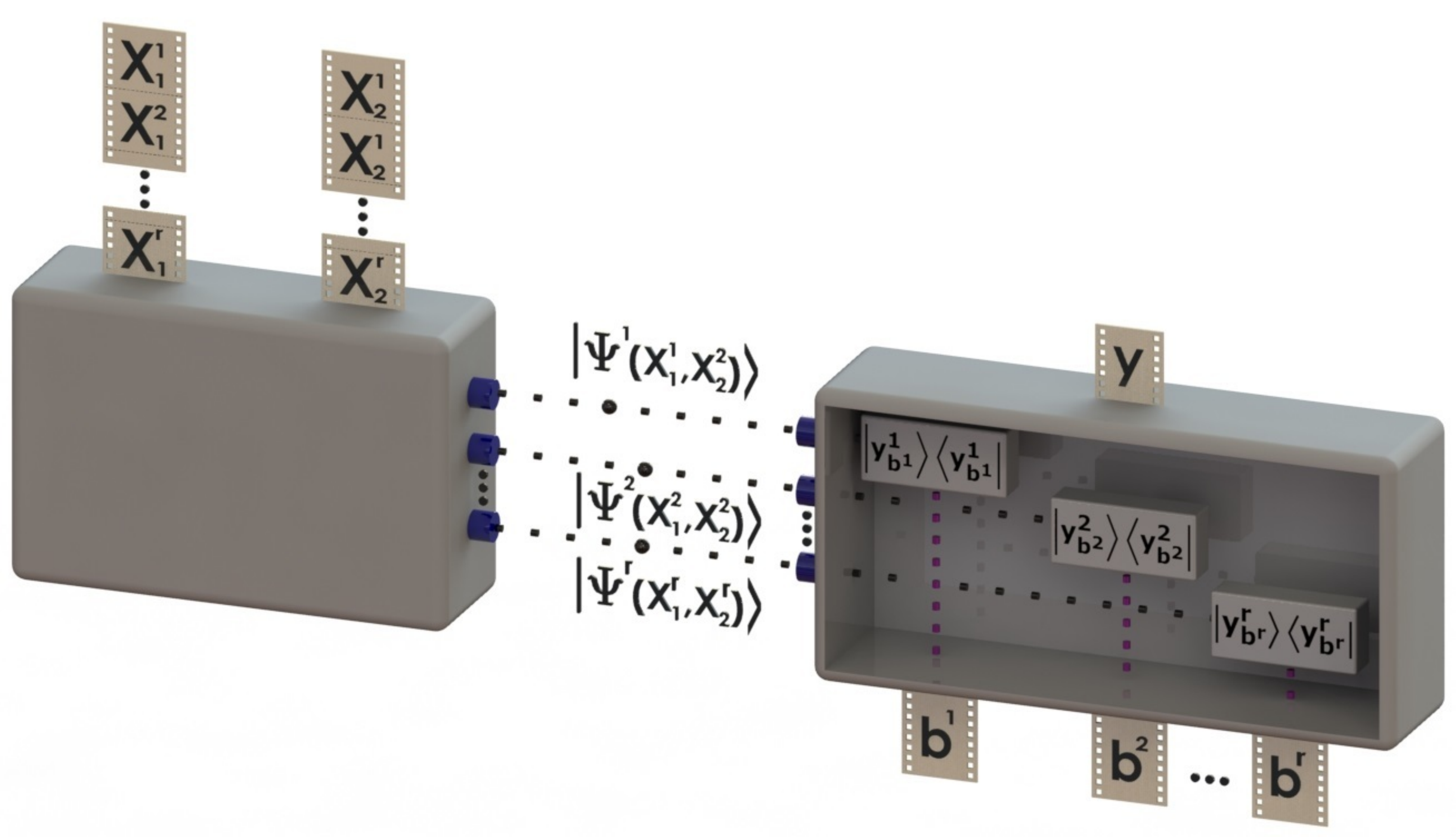}
\caption{}\label{fig:simplifiedRAC_b}
\end{subfigure}
\caption{(a) A generic QRAC with a product structure. (b) A simplified version using Lemmas \ref{lemma:id},\ref{lemma:NonAdapt}.}
\end{figure}

The first simplification we make is to show that the optimal quantum strategy does not require classical post-processing functions $\dd^y$. That is, Bob's output $b$ can simply be read out from his quantum measurements. This is typically assumed in all QRAC papers (e.g. \cite{ALMO,THMB}) but without proof.

\begin{lemma}{$~$}\label{lemma:id}\nl
Given a quantum decoding strategy $(\{M_l^y\}_l,\dd^y)$ with average success probability $\bar{p}$, there exists another quantum decoding strategy $(\{\tilde{M}_l^y\}_l,\tilde{\dd}_y)$ with average success probability $\tilde{ p }\geq \bar{p}$ and with trivial classical post processing $\tilde{\dd}_y = \text{id}$.
\end{lemma}

\begin{proof}
Let $\rho_{x}=\ee(x)$ be the states which achieve the optimal average success probability $\bar{p}$. Then Eq \eqref{eq:asp} can be expressed as:
\begin{equation}
\bar{p} = \frac{1}{nd^n} \sum_{x,y} \tr \left[ \rho_{x} \sum_{k:D_y(k)=x_y} M^y_k\right].
\end{equation}
Now, let us define new operators:
\begin{equation}
\tilde{M}^y_k = \sum_{j:D_y(j)=k } M^y_j.
\end{equation}
We can now use the same encoding states $\rho_{x_1,x_2,...,x_n}$ and write the original average success probability in terms of the new operators:
\begin{equation}\label{eq:idasp}
\bar{p} = \frac{1}{nd^n} \sum_{x,y} \tr \left[ \rho_{x}  \tilde{M}^y_{x_y} \right].
\end{equation}
Since we used a fixed encoding strategy and have a new decoding strategy, in principle we could have $\tilde{p} \geq \bar{p}$ after further optimization. Also, we see in Eq \eqref{eq:idasp} that there is no need for explicit classical post-processing (i.e. $\tilde{D}_y(k)=k$). Thus, hereafter, quantum decoding strategies will simply be written as $\{M_b^y\}_b$, since they will directly output the guess $b$.
\end{proof}

Therefore, the most general allowed measurement strategy is:

\begin{definition}{$~$}\label{def:adaptive}\nl
Assume that Bob receives $r$ states from Alice: $\rho = \rho_{\alpha_1}^1 \otimes \rho_{\alpha_2}^2 \otimes \dots \otimes \rho_{\alpha_r}^r$ (in fact, by \cite{ALMO} these could be assumed to be pure states), where each $\rho_{\alpha_i}^i \in \mathcal{S}(\CC^{d_i})$ and $d=d_1d_2 \cdots d_r$. By Lemma \ref{lemma:id}, let the measurement outcome of $\rho_{\alpha_i}^i$ be $b^i\in\{1,2,\ldots,d_i$\}. We call a \textbf{sequential adaptive strategy} any scheme where Bob uses previous measurement outputs to determine the measurement basis of future states. That is, when measuring the state $\rho_{\alpha_j}^j$, the basis $\{M_l^{y,b^1,b^2,\dots,b^{j-1}}\}_{l=1}^{d_j}$ could depend on the previously measured systems.
\end{definition}

This scenario is problematic, since optimizing sequential adaptive quantum strategies turns out to be extremely complicated in general. One of our main technical contributions is to show that they are not necessary for optimality.

\begin{lemma}{$~$}\label{lemma:NonAdapt}\nl
There exists an optimal strategy that does not use sequential adaptive measurements.
\end{lemma}
\begin{proof}
Let's assume we have a strategy that uses sequential adaptive measurements. Fix the choice
of all encoded states and measurements. Then, we show that there exists a
strategy without sequential adaptive measurements, that gives at least as high average
success probability, as the original one.
To show this, let us write the average success probability for the mentioned
sequential adaptive strategy as:
\begin{equation}
\begin{split}
\bar{p} & =\frac{1}{2d^2}\sum_{x,y}\mathbb{P}(B^1=x_y^1, B^2=x_y^2,\cdots
, B^r=x_y^r~|~X=x, Y=y) \\
& = \frac{1}{2d^2}\sum_{x,y}\mathbb{P}(B^1=\text{correct}, B^2=\text{correct}
,\cdots, B^r=\text{correct}~|~X=x, Y=y) \\
& = \mathbb{P}(B^1=\text{correct}, B^2=\text{correct},\cdots, B^r=
\text{correct}) \\
& = \mathbb{P}(B^r=\text{correct}~|~B^{r-1}=\text{correct},\cdots,
B^1=\text{correct})\mathbb{P}(B^{r-1}=\text{correct},\cdots, B^1=
\text{correct}) \\
& = \ldots = \prod_{k=r}^1\mathbb{P}(B^k=\text{correct}~|~B^{k-1}=
\text{correct},\cdots, B^1=\text{correct}),
\end{split}
\end{equation}
where we used the definition of conditional probability multiple times. By
construction, $B^k$ can only depend on such $B^j$s that $j<k$. Now, we
can use the fact, that the conditional probability is again a valid probability
measure, thus we can apply completeness of probabilities. Let us denote
$\prod_{k=r}^m\mathbb{P}(B^k=\text{correct}~|~B^{k-1}=
\text{correct},\cdots, B^1=\text{correct}) \equiv \mathcal{P}^m$. Then
\begin{equation}
\begin{split}
\bar{p} & = \mathcal{P}^3\cdot\mathbb{P}(B^2=\text{correct}~|~B^1=
\text{correct}) \mathbb{P}(B^1=\text{correct}) \\
& = \mathcal{P}^3\Big(\sum_{s=1}^{d_1}\mathbb{P}(B^2=\text{correct}~|~
B^1=\text{correct}, B^1=s)\mathbb{P}(B^1=s~|~B^1=\text{correct})\Big)
\mathbb{P}(B^1=\text{correct}) \\
& =\mathcal{P}^3\Big(\sum_{s=1}^{d_1}\mathbb{P}(B^2=\text{correct}~|~
B^1=\text{correct}, B^1=s)\mathbb{P}(B^1=s, B^1=\text{correct})\Big).
\end{split}
\end{equation}
We see that the events
$(B^1=\text{correct})$ and $(B^2=\text{correct})$ are independent when
conditioning on the value of $B^1$, i.e.
\begin{equation}\label{conditional_independence_D}
\begin{split}
& \mathbb{P}(B^2=\text{correct}, B^1=\text{correct}~|~B^1=s) \\
& =\mathbb{P}(B^1=\text{correct}~|~B^1=s)\mathbb{P}(B^2=\text{correct}
~|~B^1=s),
\end{split}
\end{equation}
for any $s\in\{1,\ldots,d_1\}$. This is because if we condition on the value
of $B^1$, we fix all the states and measurements (remember that the strategy is
fixed, and the only freedom is in the choice of measurement basis on qudit 2 (see Fig. \ref{fig:simplifiedRAC_b})).
Then, since our qudits are in a product state, their outcomes are independent.

From equation (\ref{conditional_independence_D}) it follows that
\begin{equation}
\mathbb{P}(B^2=\text{correct}~|~B^1=\text{correct}, B^1=s)=\mathbb{P}(B^2=
\text{correct}~|~B^1=s),
\end{equation}
and thus
\begin{equation}
\begin{split}
\bar{p} & = \mathcal{P}^3\Big(\sum_{s=1}^{d_1}\mathbb{P}(B^2=\text{correct}~|~
B^1=s)\mathbb{P}(B^1=s, B^1=\text{correct})\Big) \\
& \le \mathcal{P}^3\Big(\sum_{s=1}^{d_1}\mathbb{P}(B^2=\text{correct}
)\mathbb{P}(B^1=s, B^1=\text{correct})\Big) = \mathcal{P}^3\cdot
\mathbb{P}(B^2=\text{correct})\mathbb{P}(B^1=\text{correct}),
\end{split}
\end{equation}
where $\mathbb{P}(B^2=\text{correct})=\max_{s\in\{1,\ldots,d_1\}}\mathbb{P}
(B^2=\text{correct}~|~B^1=s)$, i.e. we choose the measurement basis which gives
the greatest success probability for qudit 2, hence eliminating adaptiveness on
this qudit. Now, we use the same reasoning in order to get rid of adaptiveness
on consequtive qudits. We show that this indeed works on qudit 3, and then the
idea generalizes trivially. At this point, we have that
\begin{equation}
\begin{split}
\bar{p} & = \mathcal{P}^4 \cdot\mathbb{P}(B^3=\text{correct}~|~B^2=
\text{correct}, B^1=\text{correct})\mathbb{P}(B^2=\text{correct})\mathbb{P}
(B^1=\text{correct}) \\
& = \mathcal{P}^4\Big(\sum_{s=1}^{d_2}\sum_{t=1}^{d_1}\mathbb{P}(B^3=
\text{correct}~|~B^2=\text{correct}, B^1=\text{correct}, B^2=s, B^1=t)
\\ & \times
\mathbb{P}(B^2=s~|~B^2=\text{correct})\mathbb{P}(B^1=t~|~B^1=\text{correct})
\Big) \mathbb{P}(B^2=\text{correct})\mathbb{P}(B^1=\text{correct}) \\
& = \mathcal{P}^4\Big(\sum_{s=1}^{d_2}\sum_{t=1}^{d_1}\mathbb{P}(B^3=
\text{correct}~|~B^2=\text{correct}, B^1=\text{correct}, B^2=s, B^1=t)
\\ & \times
\mathbb{P}(B^2=s, B^2=\text{correct})\mathbb{P}(B^1=t, B^1=\text{correct})
\Big)
\end{split}
\end{equation}
(here, we implicitly used the already proven fact that qudits 1 and 2 are independent of each other). Now the conditional independence goes as
\begin{equation}
\begin{split}
& \mathbb{P}(B^3=\text{correct}, B^2=\text{correct}, B^1=\text{correct}~|~
B^2=s, B^1=t) \\
& =\mathbb{P}(B^3=\text{correct}~|~B^2=s, B^1=t)\mathbb{P}(B^2=\text{correct}
, B^1=\text{correct} ~|~B^2=s, B^1=t),
\end{split}
\end{equation}
since fixing all measurement bases yields independent outcomes. From this it
follows that
\begin{equation}
\mathbb{P}(B^3=\text{correct}~|~B^2=\text{correct}, B^1=\text{correct},
B^2=s, B^1=t)
=\mathbb{P}(B^3=\text{correct}~|~ B^2=s, B^1=t),
\end{equation}
and thus
\begin{equation}
\begin{split}
\bar{p} & =
\mathcal{P}^4\Big(\sum_{s=1}^{d_2}\sum_{t=1}^{d_1}\mathbb{P}(B^3=
\text{correct}~|~B^2=s, B^1=t)
\mathbb{P}(B^2=s, B^2=\text{correct})\mathbb{P}(B^1=t, B^1=\text{correct})
\\ & \le
\mathcal{P}^4\Big(\sum_{s=1}^{d_2}\sum_{t=1}^{d_1}\mathbb{P}(B^3=
\text{correct})
\mathbb{P}(B^2=s, B^2=\text{correct})\mathbb{P}(B^1=t, B^1=\text{correct})
\Big) \\
& = \mathcal{P}^4\cdot\mathbb{P}(B^3=\text{correct})\mathbb{P}(B^2=
\text{correct})\mathbb{P}(B^1=\text{correct}),
\end{split}
\end{equation}
where $\mathbb{P}(B^3=\text{correct})=\max_{\substack{s\in\{1,\ldots,d_2\}\\
t\in\{1,\ldots,d_1\}}}{\mathbb{P}(B^3=\text{correct}~|~B^2=s, B^1=t)}$,
meaning that we choose the measurement basis that gives the greatest success
probability on qudit 3. It is clear now that this reasoning applies for all
qudits and thus
\begin{equation}
\bar{p}=\prod_{k=r}^1\mathbb{P}(B^k=\text{correct}),
\end{equation}
and it is a non-adaptive strategy.
\end{proof}

\subsection{Trade-Off Functions}

The usefulness of non-adaptive strategies is that in essence, Alice and Bob are playing $r$ QRACs in parallel (see Fig. \ref{fig:simplifiedRAC_b}). However, the optimal average success probability is not necessarily given by the independent optimal strategies on the individual subspaces. This is easily understood when one remembers that the winning condition is that $b=x_y$ as a whole, and no ``partial points" are awarded if only a part of the string is guessed correctly. Before proceeding, it is illustrative to look at the ASP once again, but written in the following way:

\begin{equation}
\begin{split}
\bar{p} &= \frac{1}{2} \left[ \frac{1}{d^2} \left( \sum_{x_1,x_2}  \PP(B=x_1|X=x_1x_2 , Y=1) \right) + \frac{1}{d^2} \left( \sum_{x_1,x_2} \PP(B=x_2|X=x_1x_2 , Y=2)     \right) \right]\\
&= \frac{1}{2} \left[ \PP(\text{Bob correctly guesses } x_1) + \PP(\text{Bob correctly guesses } x_2)     \right],
\end{split}
\end{equation} where we have defined $\PP(\text{Bob correctly guesses } x_y)$ as the average probability of success, if the $y$-th dit is asked. Let us remark that these probabilities are not independent and are clearly strategy dependent. It is this first dependency that will be our object of study:

\begin{definition}{$~$}\nl
\label{qtradeoff}
Let $z=\PP(\text{Bob correctly guesses } x_1)$. Then we define the \textbf{quantum  trade-off function} $\mathcal{M}^{q}_d (z)$ in dimension $d$ as:
\begin{equation}
\mathcal{M}^{q}_d (z) = \underset{ (\ee , \{ M^y_l\}_l)}{\max} \{ \PP(\text{Bob correctly guesses } x_2)  | \PP(\text{Bob correctly guesses } x_1)=z \} ,
\end{equation}
where the maximization is limited to all quantum encoding-decoding strategies which respect the condition of guessing $x_1$.
\end{definition}

In fact, one could formally write the optimal ASP in terms of the trade-off function as:

\begin{equation}
\label{eq:asp2}
\bar{p}_{Q_d} = \underset{z\in[\frac{1}{d},1]}{\max} \SP \frac{1}{2} \left[ z + \mathcal{M}^q_d(z) \right].
\end{equation}

We will devote a later Lemma (\ref{lemma:magic}) to investigating the functional form of the quantum $\mathcal{M}^{q}_d$. For now, we return to the problem of the $r$ QRACs in parallel. When writing out the average success probability, we have to calculate the probability that Alice and Bob win given inputs $x_1,x_2,y$. That is,

\begin{equation}
\begin{split}
\PP(B=x_y|X=x_1x_2 , Y=y) &= \PP(B^1=x^1_y , B^2=x^2_y  , \dots , B^{r} = x^{r}_y | X=x_1x_2 , Y=y) \\
&= \prod_{k=1}^{r} \PP(B^k=x^k_y|X=x_1x_2 , Y=y).
\end{split}
\end{equation}

The first equality is just expanding the dits into $r$ substrings ($B=B^1B^2\dots B^r$ and $x_y = x_y^1 x_y^2 \dots x_y^r$). To obtain the second equality, we use the fact that the QRACs are independent. According to Lemmas \ref{lemma:id} and \ref{lemma:NonAdapt}, Bob will use identity decoding on each measurement and output $b=b^1b^2\dots b^{r}$ as a guess for $x_y$. This in turn implies that the $k$th information carrier only has information about $x^k_1$ and $x^k_2$, i.e. $\PP(B^k=x^k_y|X=x_1x_2 , Y=y)$ only depends on $x^k_{1}$ and $x^k_2$.

Hence, let us define
\begin{equation}
\PP(\text{Bob correctly guesses } x^k_y) = \frac{1}{(d_k)^2} \sum_{x^k_1,x^k_2\in[d_k]} \PP(B^k=x^k_y|X^k=x_1^kx_2^k , Y=y).
\end{equation}

Then, Alice and Bob are trying to maximize the following global expression:
\begin{equation}
\label{eq:aspmagic}
\bar{p}_{Q_{d_1}Q_{d_2}\dots Q_{d_{r}}} = \underset{z^1\in[\frac{1}{d_1},1],z^2\in[\frac{1}{d_2},1],\dots, z^{r}\in[\frac{1}{d_{r}},1]}{\max} \SP\frac{1}{2} \left[ z^1  z^2  \dots z^r  + \mathcal{M}^q_{d_1}(z^1)\mathcal{M}^q_{d_2}(z^2)\dots \mathcal{M}^q_{d_{r}}(z^{r})\right].
\end{equation}

By optimizing \eqref{eq:aspmagic} , we are able to calculate the average success probability for separable states, and compare it to the optimal average success probability of \eqref{eq:asp2}. We now turn to showing the form of $\mathcal{M}^q_d(z)$.

\begin{lemma}{$~$}\label{lemma:magic}\nl
The following are equivalent forms of $\mathcal{M}^q_d(z)$:
\begin{equation}\label{eq:magic1}
\mathcal{M}^q_d(z) = 1 - \left(\frac{d-1}{d} \right) \left( \sqrt{z} - \sqrt{\frac{1-z}{d-1}}\right)^2,
\end{equation}
\begin{equation}\label{eq:magic2}
\mathcal{M}^q_d(z) = \cos^2\left( \cos^{-1}\left(\frac{1}{\sqrt{d}}\right) - \cos^{-1}\left(\sqrt{z}\right) \right).
\end{equation}
Furthermore, they are achieved when Bob's measurement bases are mutually unbiased.
\end{lemma}

\begin{proof}
Let Bob's decoding bases be $\{|\psi_k\ra\}_k$, and $\{|\phi_k\ra\}_k$, corresponding to $y=1$ and $2$, respectively. Given inputs $x_1,x_2$, Alice's best strategy is to encode a superposition of $|\psi_{x_1}\ra$ and $|\phi_{x_2}\ra$. Having any orthogonal components to these states will drop her average success probability and hence those strategies will not appear in the maximization performed for the trade-off function. Explicitly:
\begin{equation}
\ee(x)=|x \ra = \frac{1}{\sqrt{N}} \left( t|\psi_{x_1}\ra + e^{\ii \zeta}(1-t) |\phi_{x_2}\ra \right) ,
\end{equation}
where $N = 1 + 2t(1-t)\left(\Re[e^{\ii \zeta} \la \psi_{x_1} | \phi_{x_2} \ra] -1 \right)$ is a normalization factor, $t\in[0,1]$ is a parameter that will vary to change Bob's probability of correctly guessing the first dit, and $\zeta\in[0,2\pi)$ is a phase. It can be verified that $\zeta = - \text{Arg}(\la \psi_{x_1} | \phi_{x_2} \ra)$, i.e. $e^{\ii \zeta} \la \psi_{x_1} | \phi_{x_2} \ra \in \RR^+$  simultaneously maximizes both $|\la \psi_{x_1} | x \ra|^2$ and $|\la \phi_{x_2} | x \ra|^2$, for all $t\in[0,1]$. These are the probabilities of Bob correctly guessing $x_1$ and $x_2$, respectively. With this choice of $\zeta$ then:
\begin{equation}
\label{eq:guessx1}
z_{x} \equiv |\la \psi_{x_1} | x \ra|^2 = \frac{\left(t + \sqrt{s_{x}}(1-t)\right)^2}{N} ,
\end{equation}
\begin{equation}
\label{eq:guessx2}
|\la \phi_{x_2} | x \ra|^2 = \frac{\left(t\sqrt{s_{x}} + (1-t)\right)^2}{N} ,
\end{equation}
where $s_{x} = |\la \psi_{x_1} | \phi_{x_2} \ra|^2$. Inverting equation \eqref{eq:guessx1} to have $t=t(z_{x},s_{x})$:
\begin{equation}
t = \frac{-z_{x}+\sqrt{s_{x}}(\sqrt{s_{x}}+z_{x}-1)\pm \sqrt{(s_{x}-1)z_{x}(z_{x}-1)}}{(\sqrt{s_{x}}-1)(\sqrt{s_{x}}-1+2z_{x})}.
\end{equation}
Then, inserting it into \eqref{eq:guessx2} we obtain the probability of correctly guessing the second dit, as a function of the probability of correctly guessing the first ($z_x$).
\begin{equation}
\label{eq:magicpf1}
|\la \phi_{x_2} | x \ra|^2 = (1-z_x) + s_x (2z_x -1) \pm 2\sqrt{s_x(s_x -1)z_x(z_x-1)}.
\end{equation}
We take the positive sign, since we want to maximize the average success probability. Hence, we are trying to maximize the expression:
\begin{equation}
\label{eq:magicasp22}
\bar{p} = \underset{\{|\psi_k\ra\} , \{|\phi_k\ra\} }{\max} \SP \frac{1}{2d^2} \sum_{x} \left( 1 + s_x (2z_x -1) + 2\sqrt{s_x(s_x -1)z_x(z_x-1)}  \right),
\end{equation}
subject to the conditions $0\le s_x,z_x\le1$, $\sum_x s_x = d$, and $\sum_x z_x = z d^2$, where $z=\PP(\text{Bob correctly guesses }x_1)$.

The non-constant part of the above expression can be written as $\sum_xf(s_x,
z_x)$, where $f(s_x,z_x)=s_xz_x+\sqrt{s_x(1-s_x)z_x(1-z_x)}$. This sum is a
function of the
2-by-$d^2$ matrix $S=\binom{\vec{s}^T}{\vec{z}^T}$, where the $x$-th element of
the vector $\vec{s}$ ($\vec{z}$) is $s_x$ ($z_x$). Note that for any matrix $S$
satisfying the constraints on the $s_x$ and $z_x$,
\begin{equation}\label{eq:matrixmajorization}
S^\ast\equiv\begin{pmatrix}
\frac1d & \frac1d & \ldots & \frac1d \\
z & z & \ldots & z \end{pmatrix}=
S\begin{pmatrix}
\frac{1}{d^2} & \ldots & \frac{1}{d^2} \\
\vdots & ~ & \vdots \\
\frac{1}{d^2} & \ldots & \frac{1}{d^2}\end{pmatrix}.
\end{equation}
Here, the last matrix is doubly stochastic, and hence we say that any
matrix $S$ satisfying the constraints on the $s_x$ and $z_x$ majorizes
$S^\ast$ (see \cite[Definition 15.A.2]{inequalities}). But this is equivalent
(\cite[Proposition 15.A.4]{inequalities}) to the statement that $\sum_x\phi(s_x,
z_x)\le\sum_x\phi(\frac1d,z)$ for all continuous concave functions
$\phi:\mathbb{R}^2\to\mathbb{R}$. It is straightforward to show that the
function $f(s_x,z_x)$ is concave (i.e. its Hessian is negative semi-definite)
on the domain $[0,1]\times[0,1]$,
and hence, considering the above, the ASP (Eq. \eqref{eq:magicasp22}) is
maximized by $s_x=\frac1d$ and $z_x=z$ for all $x$. Substituting
these into Eq. \eqref{eq:magicpf1} we get the form of the trade-off function:
\begin{equation}
\label{eq:magicasp2}
M^q_d(z) = 1-z + \frac{2z-1}{d} + 2 \frac{\sqrt{(d-1)z(1-z)}}{d},
\end{equation}
which can be furthered simplified into \eqref{eq:magic1}.

To obtain the other form of $\mathcal{M}^q_d(z)$ we can visualize the problem geometrically, by regarding the angle $\theta$ between two state vectors $|\xi\ra$ and $|\chi\ra$ to be $\theta = \cos^{-1}\left(|\la \xi|\chi \ra|\right)$. We have shown that the trade-off function is obtained when Bob uses two mutually unbiased bases, hence the measurement vectors $|\psi_{x_1}\ra$ and $|\phi_{x_2}\ra$ have an angle of $\cos^{-1}\left(d^{-1/2}\right)$ between them. Alice's encoded state $|x\ra$ must lie on the plane of the measurement vectors and the angle between $|x\ra$ and $|\psi_{x_1}\ra$ is $\cos^{-1}\left(\sqrt{z}\right)$. The trade-off function \eqref{eq:magic2} is then obtained when we see that the angle between $|x\ra$ and $|\phi_{x_2}\ra$ is the difference of the two angles described above.
\end{proof}

Notice that in the discussion following \eqref{eq:matrixmajorization} it was shown that $s_x = |\la \psi_{x_1} | \phi_{x_2}\ra |^2 = 1/d$ for all $x$. This is precisely the MUB condition on Bob's measurements.  To arrive at Alice's optimal strategy we need to maximize \eqref{eq:aspmagic}, using the derived representation \eqref{eq:magic2} of $\mathcal{M}^q_d(z)$. The maximization can easily done by setting $\frac{d \bar{p}_{Q_d}}{dz} = 0$, to find $z_{\max}$. Explicitly:
\begeq
z_{\max} = \mathcal{M}^q_d(z_{\max}) = \frac{1}{2} \left( 1 + \frac{1}{\sqrt{d}} \right) .
\endeq
This means that the best strategy for Alice is to encode every state $|x\ra$ into an equal superposition of $|\psi_{x_1}\ra$ and $|\phi_{x_2}\ra$ in order for the success probability to be the same, no matter which basis Bob chooses to do a measurement in. We put this into a corollary:
\begin{corollary}
\label{cor:optimalstrat}
For $2^d \rw 1$ QRACs, the optimal average success probability is achieved when Bob uses two mutually unbiased bases $\left( \{|\psi_{x_1}\ra \}_{x_1} , \{|\phi_{x_2}\ra \}_{x_2} \right)$, and Alice encodes her inputs into states $|x_1x_2\ra$ which are equal superpositions of $|\psi_{x_1}\ra$ and $|\phi_{x_2}\ra$.
\end{corollary}
Note that this optimal quantum strategy for $d$-dimensional $2^d\rw 1$ QRACs has been discussed in \cite{ABMP}. The optimal encoding strategy for Alice involves encoding her state into the eigenvector corresponding to the highest eigenvalue of the operator $\left( |\psi_{x_1}\ra \la \psi_{x_1} | + | \phi_{x_2} \ra \la \phi_{x_2} | \right)$. This is the state given in Equation (2) of the main text.

For completeness, we also define the \textit{classical trade-off function} $\mathcal{M}_d^c(z)$ in an analogous way to Definition \ref{qtradeoff}, except that the maximization is done over classical encoding-decoding strategies. In fact, this function is linear:
\begin{equation}
\label{eq:classicalmagic}
\mathcal{M}_d^c(z) = \frac{d+1}{d} - z.
\end{equation}
This can easily be checked, since the optimal success probability for $2^d\rw 1$ RACs  is known to be $\bar{p}_{C_d}=(d+1)/2d$ \cite{Czechlewski18}. This success probability can be obtained by the pure coding schemes of just sending the first or second dit, and all convex combinations of these strategies would give the same maximum. See Fig. \ref{fig:tradeoff} for a visualization of the trade-off functions with varying dimensions. Note, however, that classical strategies factorize, so that we never use the trade-off functions in this setting alone, but only in conjunction with the quantum functions, e.g. if Alice is able to encode her input dits into quantum systems of dimensions $d_1, d_2, \dots, d_{r-1}$ and the rest of the information of dimension $d_r$ classically, we would have to maximize:
\begin{equation}
\label{eq:CQaspmagic}
\bar{p}_{Q_{d_1}Q_{d_2}\dots Q_{d_{r-1}}C_{d_{r}}} = \underset{z^1\in[\frac{1}{d_1},1],\dots, z^{r}\in[\frac{1}{d_{r}},1]}{\max} \SP\frac{1}{2} \left[ z^1  z^2  \cdots z^r  + \mathcal{M}^q_{d_1}(z^1)\cdots \mathcal{M}^q_{d_{r-1}}(z^{r-1}) \mathcal{M}^c_{d_{r}}(z^{r})\right].
\end{equation}

\begin{figure}[h]
\vspace{0.5cm}
\begin{center}
\includegraphics[scale=0.5]{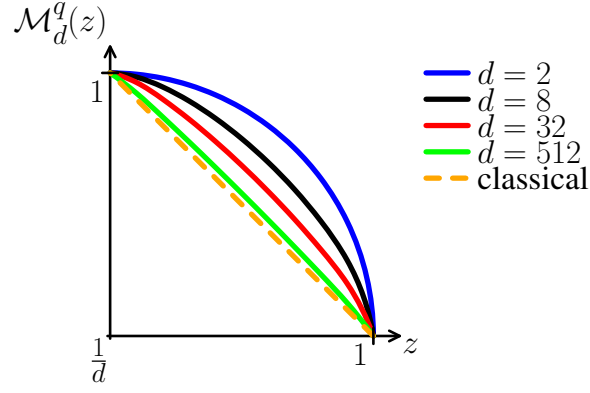}
\end{center}
\caption[]{Visualization of the quantum trade-off functions $\mathcal{M}_d^q(z)$, with varying dimensions.}
\label{fig:tradeoff}
\end{figure}

\subsection{Two Examples}
\subsubsection{d=39}
Here, we take the case $d=39$ into consideration, which will highlight the necessity of the trade-off functions. We have that $\bar{p}_{Q_{39}} = \frac{1}{2}\left(1 + \frac{1}{\sqrt{39}} \right) \approx 0.5801$. Now, we wish to know the optimal ASP if the preparation and measurement are split in terms of two systems with dimensions $d_1=13$ and $d_2=3$. Numerically we optimize \eqref{eq:aspmagic}:
\begin{equation}
\label{eq:example39}
\bar{p}_{Q_{13}Q_{3}} = \underset{z^1\in[\frac{1}{13},1],z^2\in[\frac{1}{3},1]}{\max} \SP\frac{1}{2} \left[ z^1  z^2 + \mathcal{M}^q_{13}(z^1)\mathcal{M}^q_{3}(z^2)\right] \approx 0.5217.
\end{equation}
A contour plot of of the function being maximized \eqref{eq:example39} with the maxima highlighted can be seen in Fig. \ref{fig:example}. In fact, the maximum is obtained in two different points. Let $(z^1,z^2) = (0.1944,0.4302)$ be the first point, then in fact $\left(\mathcal{M}^q_{13}(0.1944),\mathcal{M}^q_{3}(0.4302)\right)= (0.9695,0.9900)$ is the other point which achieves the maximum. The first point, where both $z^1$ and $z^2$ are relatively small, the strategy gives a strong bias to guessing the second dit $x_2$ at the expense of lowering the probability of correctly guessing the first input $x_1$. Explicitly for the first point; $\PP(\text{Bob correctly guesses } x_1)=z^1z^2 \approx 0.0836$, whereas $\PP(\text{Bob correctly guesses } x_2)=\mathcal{M}^q_{13}(z^1)\mathcal{M}^q_{3}(z^2) \approx 0.9598$. It is clear then, that the second point which achieves the maximum is just a reflection of this strategy, now giving a positive bias towards guessing the first dit.

\begin{figure}[h]
\begin{center}
\includegraphics[scale=0.5]{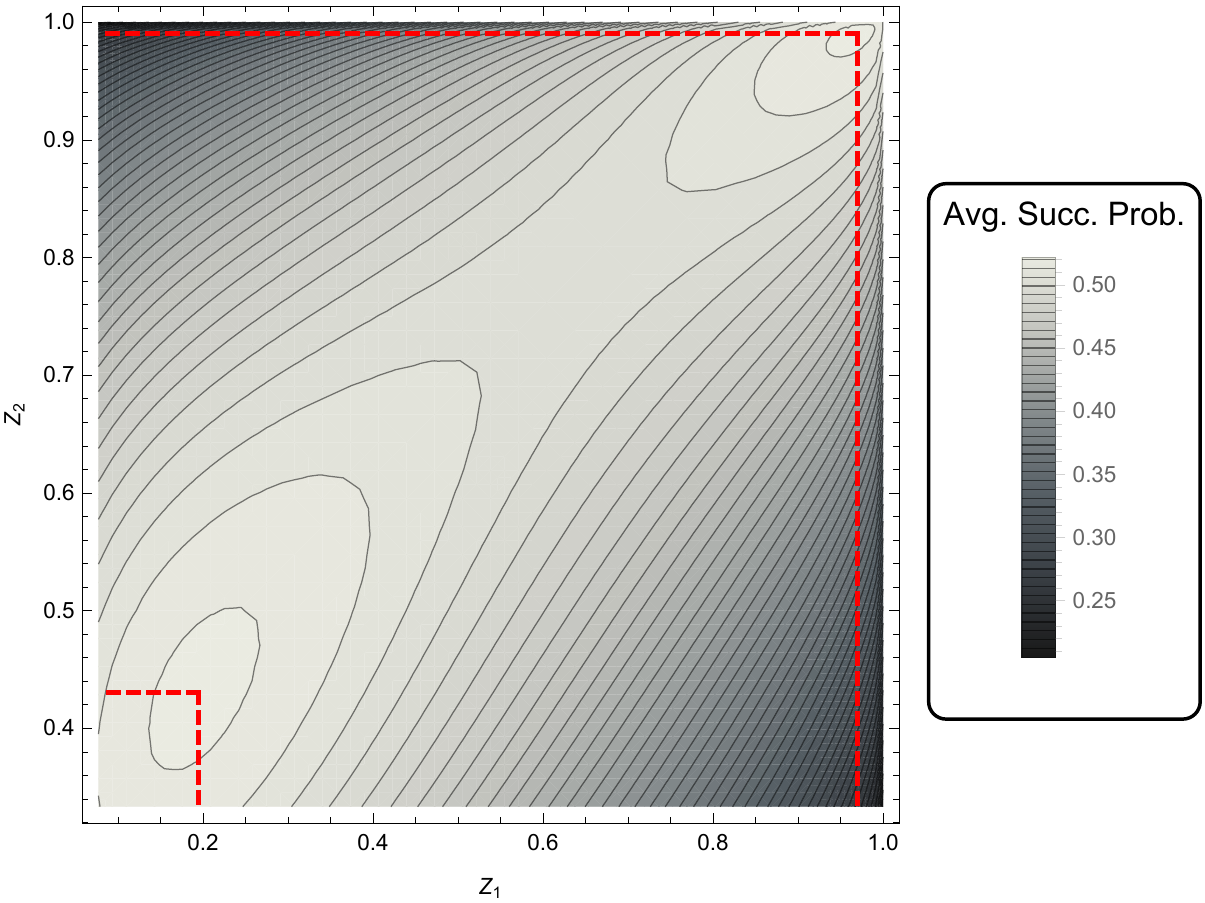}
\end{center}
\caption[]{Contour plot of \eqref{eq:example39}, for the example $d=39$. See text for details. }
\label{fig:example}
\end{figure}

To conclude, we explicitly see that $\bar{p}_{Q_{13}Q_{3}} > \bar{p}_{Q_{13}} \bar{p}_{Q_{3}} \approx 0.5037$. That is, even though Alice and Bob are using two non-interacting Hilbert spaces, the optimal strategy is a global one, instead of playing strictly independent QRACs.

\subsubsection{d=1024}
Now, we look at the case $d=1024$, the dimension we certify in our experiment. We compute the optimal success probabilities for all possible quantum partitions of a 1024-dimensional quantum system. The values were calculated using Eq. \eqref{eq:CQaspmagic}. The aim here is to show that $Q_{512}Q_2$ was the relevant bound for the experiment, and not e.g.\ $Q_{32}Q_{32}$ or any other partition. See Table \ref{table1024}. 

\begin{table}
\centering
\begin{tabular}{c@{\hspace{2cm}}c}\hline \hline
Case & Optimal $\bar{p}$ \\ \hline
 $Q_{1024} $ & 0.515625 \\
 $Q_{512}  Q_{2} $ & 0.500980 \\
 $Q_{256}  Q_{4} $ & 0.500654 \\
 $Q_{256}  Q_{2}  Q_{2} $ & 0.500653 \\
 $Q_{128}  Q_{8} $ & 0.500563 \\
 $Q_{128}  Q_{4}  Q_{2} $ & 0.500561 \\
 $Q_{128}  Q_{2}  Q_{2}  Q_{2} $ & 0.500560 \\
 $Q_{64}  Q_{16} $ & 0.500530 \\
 $Q_{64}  Q_{8}  Q_{2} $ & 0.500525 \\
 $Q_{64}  Q_{4}  Q_{4} $ & 0.500524 \\
 $Q_{64}  Q_{4}  Q_{2}  Q_{2} $ & 0.500523 \\
 $Q_{64}  Q_{2}  Q_{2}  Q_{2}  Q_{2} $ & 0.500523 \\
 $Q_{32}  Q_{32} $ & 0.500521 \\
 $Q_{32}  Q_{16}  Q_{2} $ & 0.500512 \\
 $Q_{32}  Q_{8}  Q_{4} $ & 0.500509 \\
 $Q_{32}  Q_{8}  Q_{2}  Q_{2} $ & 0.500508 \\
 $Q_{32}  Q_{4}  Q_{4}  Q_{2} $ & 0.500507 \\
 $Q_{32}  Q_{4}  Q_{2}  Q_{2}  Q_{2} $ & 0.500507 \\
 $Q_{32}  Q_{2}  Q_{2}  Q_{2}  Q_{2}  Q_{2} $ & 0.500506 \\
 $Q_{16}  Q_{16}  Q_{4} $ & 0.500505 \\
 $Q_{16}  Q_{16}  Q_{2}  Q_{2} $ & 0.500504 \\
 $Q_{16}  Q_{8}  Q_{8} $ & 0.500503 \\
 $Q_{16}  Q_{8}  Q_{4}  Q_{2} $ & 0.500501 \\
 $Q_{16}  Q_{8}  Q_{2}  Q_{2}  Q_{2} $ & 0.500501 \\
 $Q_{16}  Q_{4}  Q_{4}  Q_{4} $ & 0.500500 \\
 $Q_{16}  Q_{4}  Q_{4}  Q_{2}  Q_{2} $ & 0.500500 \\
 $Q_{16}  Q_{4}  Q_{2}  Q_{2}  Q_{2}  Q_{2} $ & 0.500499 \\
 $Q_{16}  Q_{2}  Q_{2}  Q_{2}  Q_{2}  Q_{2}  Q_{2} $ & 0.500499 \\
 $Q_{8}  Q_{8}  Q_{8}  Q_{2} $ & 0.500499 \\
 $Q_{8}  Q_{8}  Q_{4}  Q_{4} $ & 0.500498 \\
 $Q_{8}  Q_{8}  Q_{4}  Q_{2}  Q_{2} $ & 0.500498 \\
 $Q_{8}  Q_{8}  Q_{2}  Q_{2}  Q_{2}  Q_{2} $ & 0.500497 \\
 $Q_{8}  Q_{4}  Q_{4}  Q_{4}  Q_{2} $ & 0.500497 \\
 $Q_{8}  Q_{4}  Q_{4}  Q_{2}  Q_{2}  Q_{2} $ & 0.500496 \\
 $Q_{8}  Q_{4}  Q_{2}  Q_{2}  Q_{2}  Q_{2}  Q_{2} $ & 0.500496 \\
 $Q_{8}  Q_{2}  Q_{2}  Q_{2}  Q_{2}  Q_{2}  Q_{2}  Q_{2} $ & 0.500495 \\
 $Q_{4}  Q_{4}  Q_{4}  Q_{4}  Q_{4} $ & 0.500496 \\
 $Q_{4}  Q_{4}  Q_{4}  Q_{4}  Q_{2}  Q_{2} $ & 0.500495 \\
 $Q_{4}  Q_{4}  Q_{4}  Q_{2}  Q_{2}  Q_{2}  Q_{2} $ & 0.500495 \\
 $Q_{4}  Q_{4}  Q_{2}  Q_{2}  Q_{2}  Q_{2}  Q_{2}  Q_{2} $ & 0.500494 \\
 $Q_{4}  Q_{2}  Q_{2}  Q_{2}  Q_{2}  Q_{2}  Q_{2}  Q_{2}  Q_{2} $ & 0.500494 \\
 $Q_{2}  Q_{2}  Q_{2}  Q_{2}  Q_{2}  Q_{2}  Q_{2}  Q_{2}  Q_{2}  Q_{2} $ & 0.500493 \\ \hline \hline
\end{tabular}
\caption{All quantum cases for a 1024-dimensional system and the
respective optimal ASPs considering each product structure.}\label{table1024}
\end{table}

Notice that, since $\mathcal{M}^q_d (z) > \mathcal{M}^c_d (z)$, there is no need to calculate the classical-quantum partitions, as they would clearly be worse than the equivalent fully quantum partition. However, it is interesting to note that $Q_{512}C_2 > Q_{256}Q_4$.

\section{Experimental Considerations} \label{sec:Exps}

In this section, we deal with the analysis supporting our photonic experiment in dimension $d=1024$.

\subsection{Useful Representation of the MUBs}
From a theoretical point of view, any two mutually unbiased bases in dimension $d=1024$ would yield the optimal average success probability. However, in our optical setup, for simplicity it is better to consider a representation of the two MUBs which have only matrix elements given by $\pm 1$.  Thus requiring only phase-modulations of $0$ or $\pi$ to be addressed by the SLMs to encode and decode the required states. To construct such MUBs in dimension 1024, we first consider two MUBs in dimension 4:
\begin{eqnarray}
MUB^{d=4}_1&=&\frac12\left(
\begin{array}{cccc}
1 & 1 & 1 & 1\\
1 & -1 & 1 & -1 \\
1 & 1 & -1 & -1 \\
1 & -1 & -1 & 1
\end{array}
\right), \\
MUB^{d=4}_2&=&\frac12\left(
\begin{array}{cccc}
1 & -1 & 1 & 1 \\
1 & -1 & -1 & -1 \\
1 & 1 & 1 & -1 \\
-1 & -1 & 1 & -1
\end{array}
\right).
\end{eqnarray}
Now, if we consider the following tensor products:
\begin{equation}
MUB_1=(MUB_1^{d=4})^{\otimes 5}, \qquad MUB_2=(MUB_2^{d=4})^{\otimes 5},
\end{equation}
we end up with two MUBs in dimension 1024, where the columns represent the basis states.

\subsection{Single Detector Scheme}
In our photonic experiment, we are dealing with a very large dimension ($d=1024$). The protocol requires Bob to perform a full von Neumann projective measurement on one of two bases before outputting his guess $b$. In the laboratory this would translate to having $1024$ different photo-detectors associated to each of the eigenvalues of the measurement performed, which is practically impossible. However, one can simulate a full $d$-outcome projective measurement to overcome this limitation, as it has been commonly done in the field of high-dimensional quantum information processing \cite{DIDW14,QKD16,Pan2016,Fickler2012}. The basic idea is that Bob uses a flexible detector scheme, which can project the incoming state to each one of the MUBs states required in the protocol. Thus, estimating the probability for each basis state collapse individually with only one detector.

In this case, one uses an extra randomly uniform input $j\in[d]$ on Bob's side. Depending on his inputs $y,j$ Bob will measure the operators $\{|m_{j}^{y}\ra \la m_{j}^{y} | , \ONE - |m_{j}^{y}\ra \la m_{j}^{y} | \}$. If Alice's state collapses on $|m_{j}^{y}\ra \la m_{j}^{y} |$, i.e. a photon is recorded by Bob while the scheme is set to make the projection $|m_{j}^{y}\ra \la m_{j}^{y} |$, he will claim that $x_y=j$. Otherwise, he will assume that $x_y \neq j$. A full von Neumann measurement is simulated in the case that
\begin{equation}
\label{eq:assumeID}
\sum_{j\in[d]} |m_{j}^{y}\ra \la m_{j}^{y} | = \ONE , \SP \SP \forall y\in[n].
\end{equation}

Let us consider the events where $x_y=j$ and define the total number of such events $X_1$. Let us also define $D_1$ as the number of "clicks" from the experiment in those cases. Likewise, let $X_2$ denote the number of events where $x_y \neq j$, and $D_2$ the clicks in those cases. Assuming uniform sampling, $(d-1)X_1 \approx X_2$.

To get an appropriate figure of merit for the experiment in this scenario, consider first the \textit{total experimental efficiency}:
\begin{equation}
\label{eq:defnu}
\nu := \frac{\# \text{ real clicks}}{\# \text{ theoretically expected clicks}}.
\end{equation}
Note that this efficiency does not assume anything about the inner-workings of the actual experimental setup, making it still compatible with the device independent approach. Let $q$ be the average success probability of a given strategy, then:
\begin{equation}
\label{eq:nu}
\nu = \frac{D_1 + D_2}{q X_1 + \left(\frac{1-q}{d-1}\right) X_2} = \frac{D_1 + D_2}{X_1}.
\end{equation}
To calculate the number of theoretically expected clicks, we use the average failure probability $\left(\frac{1-q}{d-1}\right)$ for simplicity, but without loss of generality. Furthermore, note that the average success probability is the ratio of the times Bob correctly guessed $x_y=j$, to the number of times he should have guessed it to be $x_y=j$:
\begin{equation}
\label{eq:D1X1}
\frac{D_1}{X_1} = \nu q,
\end{equation} Then, by combining equations \eqref{eq:nu} and \eqref{eq:D1X1}, we obtain:
\begin{equation}
\label{eq:FOM}
 q = \frac{D_1}{D_1+D_2},
\end{equation}
which will be our main experimental figure of merit to calculate the average success probability $q$ of the strategy. There are several benefits of using \eqref{eq:FOM} : (1) It has an easy operational interpretation as ``fraction of times Bob clicks correctly, compared to the total number of clicks", (2) since it only uses the data from the clicks, it is more experimentally friendly, not lowering the statistics due to detector malfunction or lossy channels, (3) from how it was derived, it does not assume the inner workings of the experiment, making it quite general, and most importantly (4) with the assumption of Eq. \eqref{eq:assumeID}, it is equivalent to the standard form of the ASP, i.e., Eq.
\eqref{eq:idasp}.

\subsection{Robustness of the ASP to Detection Efficiency and Poissonian Source}
In the previous section, we arrived at \eqref{eq:D1X1} by assuming that there was only one photon present in each experimental round. However, in our experimental setup we do not have a perfect single photon source, and multi-photon events can occur. The problem with having more than one photon in the system, is that our detector does not resolve the number of detected photons (otherwise this would not be an issue, and we would simply discard events with more than one photon). The nature of our detection event, the so-called ``click", is in fact the probabilistic event ``\textit{at least 1 photon detected}". Of course this event can be understood as the complement of the event ``\textit{no photon detected}". If we assume for a brief moment that $\nu=1$, and that there is a $n$-photon event, the probability of having a ``click``-event would be:
\begeq
\PP(\text{detecting at least 1 photon }|n\text{-photon event}) = 1 - (1-q)^k .
\endeq
 
Due to the nature of laser light formation, we consider a Poisson distribution for our photon production, with mean $\mu$ which can be experimentally tuned. Now, we return to the case of having experimental efficiency $\nu$. Imagine that there are $n$ photons with Alice's state $|\Psi\ra$ present, out of which only $k$ collapse onto the correct state $|\Phi\ra$ during the measurement process, and then each of the $k$ photons have a $\nu$ probability of being detected. Hence, the probability of at least one click would be:

\begin{equation}
\frac{D_1}{X_1} = \sum_{n=1}^{\infty} \mathbb{P}(n\text{ photons produced})\sum_{k=1}^n \mathbb{P}(k \text{ of the }n \text{ photons collapsing on } |\Phi\ra |\SP n\text{-photon event})) \mathbb{P}(\text{at aeast 1 detected}) .
\end{equation}

This expression is fully general. We now explicitly introduce the Poissonian distribution:
\begin{equation}
\label{eq:pexpnumu1}
\frac{D_1}{X_1} = \sum_{n=1}^\infty \left( \frac{ \mu^n e^{-\mu} }{n!} \right)   \sum_{k=1}^n \binom{n}{k} q^k (1-q)^{n-k} \left( 1 - (1-\nu)^k \right) .
\end{equation}

To simplify matters, we look just at the inner summation to get:
\begin{equation}
\label{eq:kphoton}
\sum_{k=1}^n \binom{n}{k} q^k (1-q)^{n-k} \left( 1 - (1-\nu)^k \right) = 1 - (1-\nu q)^n .
\end{equation}
Which is what we could have intuitively guessed since the beginning. If there are $k$ photons present, then the probability to detect at least 1 photon with a $\nu$ experimental efficiency is just $1 - (1-\nu q)^n$. Then, putting \eqref{eq:kphoton} into \eqref{eq:pexpnumu1} and carrying out the sum we obtain:
\begin{equation}
\label{eq:D1X1etamu}
\frac{D_1}{X_1} = 1 - e^{- \nu \mu q} .
\end{equation}
We note that while deriving this, we have been assuming the optimal QRAC strategy for the encoded states and measurement operators. In particular, $q$ does not depend on the inputs of Alice and Bob, (as shown in lemma \ref{lemma:magic}), i.e.\ every round performs the same as the average. In the same way, the average failing probability $\left(\frac{1-q}{d-1}\right)$ will be modified as:
\begeq
\label{eq:D2X2etamu}
\frac{D_2}{X_2} = 1 - e^{-  \nu \mu \left(\frac{1-q}{d-1}\right)} .
\endeq
 Then, if we divide the rhs of \eqref{eq:FOM} by $X_1$, and we use \eqref{eq:D1X1etamu} and \eqref{eq:D2X2etamu}, we obtain:
\begeq
\label{eq:FOMetamu}
 \frac{D_1}{D_1+D_2} = \frac{1 - e^{- \nu \mu 	q}}{1 - e^{- \nu \mu 	q} +(d-1)\left(1 - e^{-  \nu \mu \left(\frac{1-q}{d-1}\right)}\right)} ,
\endeq
which relates the theoretical average success probability of the strategy $q$, to our experimental figure of merit. We interpret this as follows: suppose Alice and Bob's strategy predicts an average success probability of $q$, and we experimentally know the value $\nu\mu$. Then, equation \eqref{eq:FOMetamu} gives the maximally allowed value of the figure of merit, assuming no other experimental errors. Experimentally, this allows us to fine-tune the $\mu$ parameter, to be sure the $Q_{512}Q_{2}$ value can be violated.

The first order term of \eqref{eq:FOMetamu} in the small parameter $\nu \mu$ ($0.052$ in our setup) is:
\begeq
\label{eq:FOMetamuseries}
\frac{D_1}{D_1+D_2} = q - \frac{1}{2} \left( \frac{1-q}{d-1} \right) q (dq-1) \nu \mu + O\left( (\nu \mu)^2 \right) .
\endeq

\bibliographystyle{apsrev4-1}
\bibliography{cqrefs}

\end{document}